\documentclass[a4paper,12pt,oneside,reqno]{amsart}
\usepackage[latin1]{inputenc}
\usepackage{hyperref}
\usepackage[headinclude,DIV13]{typearea}
\areaset{16cm}{24.7cm}
\usepackage{txfonts,amssymb,amsmath,amsthm,bbm,wasysym}
\usepackage[pdftex]{graphicx}
\usepackage{geometry}
\usepackage{pgfplots}
\pgfplotsset{compat=1.6}
\usepackage[framemethod=tikz]{mdframed}
\usepackage{cancel} 
\usepackage{graphicx} 
\usepackage{caption}
\usepackage{subcaption}
\usepackage{mathrsfs}
\usepackage{xcolor}
\usepackage{enumerate}
\usepackage{dsfont}
\usepackage{nicefrac}	
\theoremstyle{plain}
\newtheorem{theorem}{Theorem}[section]
\newtheorem{lemma}[theorem]{Lemma}

\newtheorem{corollary}[theorem]{Corollary}

\theoremstyle{definition}
\newtheorem{definition}[theorem]{Definition}

{\itshape}{\rmfamily}
{\itshape}{\rmfamily}

\theoremstyle{remark}
\newtheorem{remark}{Remark}

\def\paragraph#1{\noindent \textbf{#1}}

\numberwithin{equation}{section}

\def\<{\langle}
\def\>{\rangle}

\def\e{\e}

\def\f{\phi}

\def\R{{\Bbb R}}  
\def\N{{\Bbb N}}  %

\let\cal=\mathcal

\def\MM{{\cal M}}

 \def \e {{\epsilon}}

 \def \x {{\xi}}
 
 \def \f {{\phi}}

 \newcommand{\be}{\begin{equation}}
 \newcommand{\ee}{\end{equation}}

\newcommand{\bea}{\begin{eqnarray}}
 \newcommand{\eea}{\end{eqnarray}}
 
\def\TH(#1){\label{#1}}\def\thv(#1){\ref{#1}}
\def\Eq(#1){\label{#1}}\def\eqv(#1){(\ref{#1})}

 \def \1{\mathbbm{1}}

\def\Hn{\mathbb{H}^n}
\def\eps{\varepsilon}
\def\x{\textbf{x}}
\def\y{\textbf{y}}

\def\ee{\mathrm{e}}
\DeclareMathOperator*{\argmin}{arg\,min}
\DeclareMathOperator*{\argmax}{arg\,max}
\newcommand{\norm}[1]{\left\lVert#1\right\rVert}
\renewcommand{\clearpage}{}

\begin{document}
 \title[From Adaptive Dynamics to Adaptive Walks]{From Adaptive Dynamics to Adaptive Walks}
  
\author[]{Anna Kraut, Anton Bovier}
\address{A. Kraut\\ Institut f\"ur Angewandte Mathematik\\
Rheinische Friedrich-Wilhelms-Universit\"at\\ Endenicher Allee 60\\ 53115 Bonn, Germany}
\email{kraut@iam.uni-bonn.de}
\address{A. Bovier\\ Institut f\"ur Angewandte Mathematik\\
Rheinische Friedrich-Wilhelms-Universit\"at\\ Endenicher Allee 60\\ 53115 Bonn, Germany}
\email{bovier@uni-bonn.de}



\thanks{The research in this paper is partially supported by the German Research Foundation in the Priority Programme  1590 ``Probabilistic Structures in Evolution'' and  
 under Germany's Excellence Strategy -- EXC2151 -- 390873048 and EXC 2047.}
\begin{abstract}
We consider an asexually reproducing population on a finite type space whose evolution is driven by exponential birth, death and competition rates, as well as the possibility of mutation at a birth event. On the individual-based level this population can be modelled as a measure-valued Markov process. Multiple variations of this system have been studied in the simultaneous limit of large populations and rare mutations, where the regime is chosen such that mutations are separated.
We consider the deterministic system, resulting from the large population limit, and then let the mutation probability tend to zero. This corresponds to a much higher frequency of mutations, where multiple microscopic types are present at the same time. The limiting process resembles an adaptive walk or flight and jumps between different equilibria of coexisting types.
The graph structure on the type space, determined by the possibilities to mutate, plays an important role in defining this jump process.
In a variation of the above model, where the radius in which mutants can be spread is limited, we study the possibility of crossing valleys in the fitness landscape and derive different kinds of limiting walks.
\end{abstract}

\maketitle

\section{Introduction}

The concept of \textit{adaptive dynamics} is a heuristic biological theory for the evolution of a population made up of different types that has been developed in the 1990s, see  \cite{MeGeMe96,DiLa96,BoPa97,BoPa99,DiLa00}. It assumes asexual, clonal reproduction with the possibility of mutation. These mutations are rare and new types can initially be neglected, but selection acts fast and the population is assumed to always be at equilibrium. This implies a separation of the fast ecological and slow evolutionary time scale. Fixation or extinction of a mutant are determined by its \textit{invasion fitness} that describes its exponential growth rate in a population at equilibrium. This notion of fitness is dependent on the current resident population and therefore changes over time.
The equilibria do not need to be monomorphic and allow for coexistence and evolutionary branching. Eventually, so-called \textit{evolutionary stable states} can be reached, where all possible mutants have negative invasion fitness and therefore the state of the population is final.

A special case of adaptive dynamics are so-called \textit{adaptive walks} or \textit{adaptive flights}. The concept of adaptive walks was introduced by \cite{MaySmith62,MaySmith70} and further developed by \cite{KaLe87,Kau92,Orr03}. Here, evolution is modelled as a random walk on the type space that moves towards higher fitness as the population adapts to its environment. More precisely, a discrete state space is equipped with a graph structure that marks the possibility of mutation between neighbours. A fixed, but possibly random, fitness landscape is imposed on the type space. In contrast to the above, this \textit{individual fitness} is not dependent on the current state of the population. Adaptive walks move along neighbours of increasing fitness, according to some transition law, towards a local or global optimum. Adaptive flights, a term that has been introduced by \cite{KrugNeid2011}, can take larger steps and jump between local fitness maxima to eventually attain a global maximum.
Quantities of interest are, among others, the typical length of an adaptive walk before reaching a local fitness maximum and the distribution of maxima, see \cite{NoKr15}, as well as the number of accessible paths, see \cite{SchKr14,BeBruShi16,BeBruShi17}. They have been studied under various assumptions on the correlations of the fitness landscape and the transition law of the walk. Examples, mentioned by \cite{NoKr15}, are the \textit{natural adaptive walk}, where the transition probabilities are proportional to the increase in fitness, or the \textit{greedy adaptive walk}, which always jumps to the fittest available neighbour.

Over the last years, stochastic individual-based models have been introduced to study different aspects 
of evolution. They start out with a model that considers a collection of individuals. Each individual is 
characterised by a type, for example its genotype. The population evolves in time under the mechanisms of birth, death, and mutation, where the parameters depend on the types. The population size is not fixed but the resources of the environment, represented by the \textit{carrying capacity} $K$, are limited. This results in a competitive interaction between the individuals, which limits the population size to the order of $K$. The dynamics are modelled as a continuous time Markov process, as shown by \cite{FoMe04}. It is of particular interest to study the convergence of this process in the limits of large populations, rare mutations, and small mutation steps.

For a finite type space, \cite{EtKu86} have shown  that, rescaling the population by $K$, the process converges to the deterministic solution of a system of differential equations in the limit of large populations, i.e.\ as $K$ tends to infinity. The differential equations are of Lotka-Volterra type with additional terms for the effects of mutation. This result was generalised for types in $\R^d$ by \cite{FoMe04}. For finite times, in the limit of rare mutations, this deterministic system converges to the corresponding mutation-free Lotka-Volterra system. Under certain conditions, these converge in time to unique equilibrium configurations, see \cite{HofSig98,ChJaRa10}.

\cite{Cha06,ChFeMe08,ChMe11} and others have considered the simultaneous limit of large populations and rare mutations . Here, the mutation probability $\mu_K$ tends to zero as $K$ tends to infinity. They make strong assumptions on the scaling of $\mu_K$, where only very small mutation probabilities $\mu_K\ll 1/(K\log K)$ are considered. This ensures the separation of different mutation events. With high probability, a mutant either dies out or fixates in the resident population before the next mutation occurs. To balance the rare mutations, time is rescaled by $1/(K\mu_K)$, which corresponds to the average time until a mutation occurs. The limiting process is a Markov jump process called \textit{trait substitution sequence} (TSS) or \textit{polymorphic evolution sequence} (PES), depending on whether the population stays monomorphic or branches into several coexisting types. In the framework of adaptive walks, these sequences correspond to the natural walk, mentioned above.

Similar convergence results have been shown for many variations of the original individual-based model under the same scaling, including small mutational effects, fast phenotypic switches, spatial aspects, and also diploid organisms, see, e.g. 
\cite{BaBoCh17,BaBo18,ChMe07,Tran08,Lem16,CoMeMe13,NeuBo17,BoCoNeu18}. 

The drawback of all these results is the strong assumption on the mutation rate. The separation of mutations which results in small mutational effects and slow evolution has been criticised by  \cite{BaPo05}. We therefore consider a scenario where the mutation rate is much higher, although decreasing, and the mutation events are no longer separated. This allows for several mutations to accumulate before a new type fully invades the population. To study the extreme case, as first done  
by \cite{BoWa13} and recently by \cite{BoCoSm18}, we consider the two limits
 separately. We take the deterministic model, arising from limit of large populations, and let the mutation
  rate $\mu$ tend to zero while rescaling the time by $\ln1/\mu$. This corresponds to the time that a 
  mutant takes to reach a macroscopic population size of order $1$, rather than the time until a mutant 
  appears, as before. The time that the system takes to re-equilibrate is negligible on the chosen time 
  scale and hence the resulting limit is a jump process between metastable equilibrium states. 

We consider a finite type space with a graph structure representing the possibility of mutation. First, we prove that, under certain assumptions, the deterministic model converges pointwise to a deterministic jump process in the rare mutation limit. This process jumps between Lotka-Volterra equilibria of the current macroscopic types. For a (possibly polymorphic) resident population, we have to carefully track the growth of the different microscopic mutants that compete to invade the population. The first mutant to reach a macroscopically visible population size solves an optimisation problem and balances high invasion fitness and large initial conditions, where the latter is determined by the graph distance to the resident types. The limiting process can be fully described by its jump times and jump chain, which are closely related to this optimisation problem. It can make arbitrarily large jumps and may reach an evolutionary stable state.

Second, we show how we can derive different limiting processes by changing the parameters of the system. On one hand, assuming equal competition between all individuals and monomorphic initial conditions, the description of the jump process can be simplified. In this case, the invasion fitness of a type is just the difference between its own individual fitness, defined by its birth and death rate, and that of the resident type. Hence, we can relate back to the classical notion of fixed fitness landscapes in the context of adaptive walks. The limiting process resembles an adaptive flight since it always jumps to types of higher individual fitness, eventually reaching a global fitness maximum. A similar scenario was studied in the context of adaptive walks and flights by \cite{KarlKrug03,JainKrug05,JainKrug07,Jain07}. Here the fitness is also assumed to be fixed but time steps are discrete. As in our case, the transitions between macroscopic types are determined by balancing high initial conditions, depending on the distance in the type space, and high fitness.

On the other hand, we modify the deterministic system such that the subpopulations can only reproduce when their size lies above a certain threshold. This limits the radius in which a resident population can foster mutants. A threshold of $\mu^\ell$ mimics the scaling of $\mu_K\approx K^{-1/\ell}$ in the simultaneous limit, where resident types can produce mutants in a radius of $\ell$. \cite{BoCoSm18} and Champagnat, M\'el\'eard, and Tran (preprint 2019) recently studied this scaling for the type space of a discrete line. A similar scaling has also been applied to a Moran-type model by \cite{DuMay11} and an adaptive walk-type model with restricted mutation radius has been studied by \cite{JainKrug07}. The resulting limit processes of the modified deterministic system are similar to the previously mentioned greedy adaptive walk. However, they are not all restricted to jumping to direct neighbours only, and thus can cross valleys in the fitness landscape and reach a global fitness maximum. Only when we choose the extreme case $\ell=1$, the resulting limit is exactly the greedy adaptive walk.

The remainder of this paper is organised as follows.
In Section 2, we introduce the deterministic system and the corresponding mutation free Lotka-Volterra system and present the main theorems, stating the convergence to different jump processes in the limit of rare mutation for different scenarios. We relate the deterministic system to the individual-based stochastic model and present a modification that mimics the simultaneous limit of large populations and rare, but still overlapping, mutations. Moreover, we give a short outline of the strategy of the proofs.
Section 3 and 4 are devoted to the proof of the first convergence result. The proof is split into three parts. The analysis of the exponential growth phase of the mutants, which follows ideas from \cite{BoWa13}, is given in Section 3. The following Lotka-Volterra invasion phase has been studied in detail by \cite{ChJaRa10}.
In Section 4, we show how to combine the two phases to prove the main result.
Next, in Section 5, we consider the special case of equal competition, where we can simplify the description of the limiting jump process. Since the assumptions of the result from \cite{ChJaRa10} are no longer satisfied, we have to slightly change the proof.
In Section 6, we finally present an extension of the original deterministic system, where we limit the range of mutation to mimic the scaling of $\mu_K\approx K^{-1/\ell}$ in the simultaneous 
limit. In the extreme case, where only resident types can foster mutants, the greedy adaptive walk arises in the limit. For the intermediate cases, we present some first results on accessibility of types.


\section{Model introduction and main results}

In this section we introduce the deterministic model for evolution that is the focus of our studies. Similar models have been studied by \cite{HofSig98}, who give an extensive overview of models of population dynamics in their book. We present the main result of convergence of this deterministic process in the limit of rare mutations on a divergent time scale. The limiting process is a deterministic jump process that jumps between Lotka-Volterra equilibria, involving different types. In the special case of equal competition, we derive a simplified description of this limiting process. Moreover, we relate the model to the stochastic individual-based model introduced by \cite{FoMe04} and present a modification of the deterministic system that mimics the simultaneous limit of large populations and rare, but not too rare, mutations. In the case where only neighbouring types of the current resident type can arise as mutants, the limiting object is a true adaptive walk. At the end of the section we outline the proofs that are given in the following sections.

\subsection{The deterministic system and relations to Lotka-Volterra systems}

The model we consider is a classical Lotka-Volterra system with additional mutation terms. We consider a population consisting of subpopulations that are characterised by their types (e.g.\ geno- or phenotypes). In this paper we choose the $n$-dimensional hypercube $\Hn:=\{0,1\}^n$ as our \textit{type space}. The sequences of ones and zeros can, for example, be interpreted as sequences of loci with different alleles. The type $(0,...,0)$ can be seen as the wildtype while all other types have accumulated mutations on some loci. However, we will not assume to always start out with a monomorphic population of this type.

The choice of $\Hn$ can easily be generalised to any finite set. We comment on this in Section \ref{1stThm}.

The state of the system is described by $\xi^\mu_t=(\xi^\mu_t(x))_{x\in\Hn}$, where $\xi^\mu_t(x)$ denotes the size of the subpopulation of type $x$ at time $t$. $\xi^\mu_t$ can be seen as a non-negative vector or (not necessarily normalised) measure on $\Hn$.

The dynamics of $(\xi^\mu_t)_{t\geq0}$ are determined by the system of differential equations
\begin{align}
\tfrac{d}{dt}\xi^\mu_t(x)=&\left[b(x)-d(x)-\sum_{y\in\Hn}\alpha(x,y)\xi^\mu_t(y)\right]\xi^\mu_t(x)\notag\\
&+\mu\sum_{y\in\Hn}\xi^\mu_t(y)b(y)m(y,x)-\mu\xi^\mu_t(x)b(x)\sum_{y\in\Hn}m(x,y),\quad x\in\Hn,\label{prelDE}
\end{align}
where the parameters are chosen as follows.
\begin{definition}
For $x,y\in\Hn$, we define
\begin{enumerate}[-]
\item $b(x)\in\R_+$, the \textit{birth rate} of an individual with type $x$,
\item $d(x)\in\R_+$, the \textit{(natural) death rate} of an individual with type $x$,
\item $\alpha(x,y)\in\R_+$, the \textit{competitive pressure} that is imposed upon an individual with type $x$ by an individual with type $y$,
\item $\mu\in[0,1]$, the \textit{probability of mutation} at a birth event,
\item $m(x,\cdot)\in\mathcal{M}_p(\Hn)$, the \textit{law of the mutant}.
\end{enumerate}
Here $\mathcal{M}_p(\Hn)$ is the set of probability measures on $\Hn$. We assume that $m(x,x)=0$, for every $x\in\Hn$. For each $x\in\Hn$, we define $r(x):=b(x)-d(x)$, its \textit{individual fitness}.
\end{definition}
Abiotic factors like temperature, chemical milieu, or other environmental properties enter through $b$ and $d$, while biotic factors such as competition due to limited food supplies, segregated toxins, or predator-prey relationships are reflected in the competition kernel $\alpha$.

We could also let the probability of mutation depend on $x\in\Hn$ in a way such that it is still proportional to some $\mu$, i.e.\ $\mu M(x)$. However, this would not change the limiting process, therefore we stick with a constant $\mu$ for simplicity of notation.

Note that the competition term ensures that solutions are always bounded. This implies Lipschitz continuity for the coefficients, and hence the classical theory for ordinary differential equations ensures existence, uniqueness, and continuity in $t$ of such solutions $\xi^\mu_t$. Moreover, for non-negative initial condition $\xi^\mu_0$, $\xi^\mu_t$ is non-negative at all times.

\begin{definition}
For $x\in\Hn$, we denote by $|x|:=\sum_{i=1}^nx_i$ the 1-norm. We write $x\sim y$ if $x$ and $y$ are direct neighbours on the hypercube, i.e.\ if $|x-y|=1$. Else, we write $x\nsim y$. We denote the standard Euclidean norm by $\norm{\cdot}$.
\end{definition}
To ensure that the mutants which a type $x\in\Hn$ can produce are exactly its direct neighbours, we introduce the following assumption. It corresponds to only allowing single mutations.
\begin{itemize}
\item[\textbf{(A)}] For every $x,y\in\Hn$, $m(x,y)>0$ if and only if $x\sim y$.
\end{itemize}
Again, this assumption is not necessary and can easily be relaxed. However, it simplifies notation and does not change the method of the proofs. We comment on the case of general finite (directed) graphs as type spaces in Section \ref{1stThm}.

Under the above assumption, (\ref{prelDE}) reduces to
\begin{align}
\tfrac{d}{dt}\xi^\mu_t(x)=\left[r(x)-\sum_{y\in\Hn}\alpha(x,y)\xi^\mu_t(y)\right]\xi^\mu_t(x)+\mu\sum_{y\sim x}b(y)m(y,x)\xi^\mu_t(y)-\mu b(x)\xi^\mu_t(x).\label{DE}
\end{align}

In the mutation-free case, where $\mu=0$, the equations take the form of a competitive Lotka-Volterra system
\begin{align}
\tfrac{d}{dt}\xi^0_t(x)=\left[r(x)-\sum_{y\in\Hn}\alpha(x,y)\xi^0_t(y)\right]\xi^0_t(x).\label{LV}
\end{align}
Understanding this system is essential since it determines the short term dynamics of the system with mutation as $\mu\to0$. For a subset of types we study the stable states of the Lotka-Volterra system involving these types.

\begin{definition}
For a subset $\x\subset\Hn$ we define the set of \textit{Lotka-Volterra equilibria} by
\begin{align}
\text{LVE}(\x):=\left\{\xi\in(\R_{\geq0})^\x:\forall\ x\in\x:\ \Big[r(x)-\sum_{y\in\x}\alpha(x,y)\xi(y)\Big]\xi(x)=0\right\}.\label{equil}
\end{align}
Moreover, we let $\text{LVE}_+(\x):=\text{LVE}(\x)\cap(\R_{>0})^\x$. If $\text{LVE}_+(\x)$ contains exactly one element, we denote it by
$\bar{\xi}_\x$, the \textit{equilibrium size} of a population of coexisting types $\x$.
\end{definition}
\begin{remark}
If $\text{LVE}_+(\x)=\{\bar{\xi}_\x\}$, this implies $r(x)>0$ for all $x\in\x$. In the case where $\x=\{x\}$, we obtain $\bar{\xi}_x(x):=\bar{\xi}_\x(x)=\frac{r(x)}{\alpha(x,x)}$. 
\end{remark}

The following assumption ensures that for a subset $\x\subset\Hn$, such that $r(x)>0$ for all $x\in\x$, there exists a unique asymptotically stable equilibrium of the Lotka-Volterra system involving types $\x$.

\begin{itemize}
\item[\textbf{(B$_\x$)}] 
There exist $\theta_x>0$, $x\in\x$, such that
\begin{align}
\forall\ x,y\in\x:&\ \theta_x\alpha(x,y)=\theta_y\alpha(y,x),\label{sym}\\
\forall\ u\in\R^\x\backslash\{0\}:&\ \sum_{x,y\in\x} \theta_x\alpha(x,y)u(x)u(y)>0.\label{preposdef}
\end{align}
\end{itemize}

This is, for example, trivially satisfied by any symmetric, positive definite matrix $(\alpha(x,y))_{x,y\in\x}$. Under this condition, Champagnat, Jabin, and Raoul have proven convergence to a unique stable equilibrium.
\begin{theorem}[\cite{ChJaRa10}, Prop.1]\label{LVThm}
Assume (B$_\x$) for a subset $\x\subset\Hn$ such that $r(x)>0$, for all $x\in\x$.
Then there exists a unique $\bar{\xi}_\x\in(\R_+)^\x\backslash\{0\}$ such that for any solution $\xi^0_t$ to (\ref{LV}) with initial condition $\xi^0_0\in(\R_{>0})^\x\times\{0\}^{\Hn\backslash\x}$,
\begin{align}
\left.\xi^0_t\right|_\x\to\bar{\xi}_\x\,\text{ as }t\to\infty.
\end{align}
\end{theorem}

The proof of this theorem uses the Lyapunov functional
\begin{align}
L(\xi)=\frac{1}{2}\sum_{x,y\in\x} \theta_x\alpha(x,y)\xi(x)\xi(y)-\sum_{x\in\x} \theta_xr(x)\xi(x),\quad \xi\in\R^\x.
\end{align}
(\ref{sym}) ensures that
\begin{align}
\frac{d}{dt}L(\left.\xi^0_t\right|_\x)=(\nabla L)(\left.\xi^0_t\right|_\x)\cdot\tfrac{d}{dt}\left.\xi^0_t\right|_\x=-\sum_{x\in\x} \theta_x\left[r(x)-\sum_{y\in\x}\alpha(x,y)\xi^0_t(y)\right]^2\xi^0_t(x)\leq 0,
\end{align}
while (\ref{preposdef}) gives convexity of $L$.

\begin{remark}
Note that \ref{preposdef} implies
\begin{align}
\forall\ u\in\R^\x\backslash\{0\}:&\ \sum_{x,y\in\x} \theta_x\alpha(x,y)u(x)u(y)\geq \kappa_\x\norm{u}^2,\label{posdef}
\end{align}
where
\begin{align}
\kappa_\x:=\min_{u:\norm{u}=1}\sum_{x,y\in\x} \theta_x\alpha(x,y)u(x)u(y)>0.
\end{align}
We set $\kappa:=\min_{\x\subset\Hn}\kappa_\x$.
\end{remark}

Connected to this positive definiteness property and the Lotka-Volterra equilibria, we define a norm that is used to measure the distance between the current state of the population and the equilibrium size. Since the $\theta_x$, $x\in\x$, in (B$_\x$) are not unique, we fix an arbitrary choice of such parameters. In the case where $(\alpha(x,y))_{x,y\in\x}$ is irreducible, we can choose the unique normalised version where $\sum_{x\in\x}\theta_x=1$.
\begin{definition}
For $\x\subset\Hn$ such that $\text{LVE}_+(\x)=\{\bar{\xi}_\x\}$ and (B$_\x$) is satisfied, we define a scalar product on $\R^\x$ (or $\MM(\x)$, the set of non-negative measures on $\x$) by
\begin{align}
\langle u,v\rangle_\x:=\sum_{x\in\x}\frac{\theta_x}{\bar{\xi}_\x(x)}u(x)v(x),\quad u,v\in\R^\x.\label{defnorm}
\end{align}
The corresponding norm is defined by $\norm{u}_\x:=\sqrt{\langle u,u\rangle_\x}$.
\end{definition}

This scalar product is chosen exactly in a way such that we can use the positive definiteness (\ref{posdef}) and the properties of $\bar{\xi}_\x$. Moreover, we notice that
\begin{align}
c_\x^2\norm{u}^2:=\left(\min_{x\in\x}\frac{\theta_x}{\bar{\xi}_\x(x)}\right)\norm{u}^2\leq\norm{u}^2_\x\leq \left(\max_{x\in\x}\frac{\theta_x}{\bar{\xi}_\x(x)}\right)\norm{u}^2=:C_\x^2\norm{u}^2.
\end{align}

\begin{remark}
Throughout the paper, constants labelled $c$ and $C$ have varying values. Specific constants, as $c_\x$ and $C_\x$ above, are labelled differently and referenced when used repetitively.
\end{remark}

While some types $\x$ coexist at their equilibrium size $\bar{\xi}_\x$, other types $y\in\Hn\backslash\x$, which only have a small population size, grow in their presence. Considering the rate of exponential growth in (\ref{LV}), we formulate a notion of invasion fitness.
\begin{definition}
For $\x\subset\Hn$ such that $\text{LVE}_+(\x)=\{\bar{\xi}_\x\}$ and $y\in\Hn$, we define the \textit{invasion fitness} of an individual with type $y$ in a population of coexisting types $\x$ at equilibrium by \linebreak $f_{y,\x}:=r(y)-\sum_{x\in\x}\alpha(y,x)\bar{\xi}_\x(x)$.
\end{definition}

Notice that $f_{x,\x}=0$ for all $x\in\x$. In contrast to the individual fitness $r$, which is fixed, this notion of fitness varies over time and depends on the current resident types.


\subsection{Convergence to a deterministic jump process}\label{1stThm}

We now come back to the system (\ref{DE}), involving mutation. We assume that the system starts out close to the equilibrium size of some subset of types $\x\subset\Hn$ and study its evolution over time. We distinguish between macroscopic resident types that coexist at their equilibrium size and microscopic mutant types that have a population size that tends to $0$ as $\mu\to0$. The initial conditions are specified as follows.

\begin{definition}
A collection of measures $\xi^\mu_0\in\MM(\Hn)$, depending on $\mu$, \textit{satisfies the initial conditions for resident types $\x\subset\Hn$, $\eta>0$, and $\bar{c}>0$} if $\text{LVE}_+(\x)=\{\bar{\xi}_\x\}$ and there exists a $\mu_0\in(0,1]$ and constants $0\leq c_y\leq C_y<\infty$ and $\lambda_y\geq0$, for each $y\in\Hn$, such that, for every $\mu\in(0,\mu_0]$,
\begin{align}
\xi^\mu_0(y)\in[c_y\mu^{\lambda_y},C_y\mu^{\lambda_y}],\label{initial}
\end{align}
where
\begin{align}
\forall~y\in\x:&\ \lambda_y=0,\ \bar{\xi}_\x(y)-\eta\frac{\bar{c}}{\sqrt{|\x|}}\leq c_y,C_y\leq\bar{\xi}_\x(y)+\eta\frac{\bar{c}}{\sqrt{|\x|}},\label{initialx}\\
\forall~y\in\Hn\backslash\x:&\ \lambda_y>0,\ 0\leq c_y,C_y<\infty\quad\text{or}\\
&\ \lambda_y=0,\ 0\leq c_y,C_y\leq\frac{\eta}{3},\ f_{y,\x}<0.
\end{align}
If $\xi^\mu_0(y)\equiv0$, we choose any $\lambda_y>\max_{z\in\Hn: \xi^\mu_0(z)>0}\lambda_z+n$.\\
We write $\xi^\mu_0\in\text{IC}(\x,\eta,\bar{c})$.
\end{definition}

This definition is very technical. We could choose more simple initial conditions for our main theorem, for example a monomorphic macroscopic type and no microscopic types of positive population size. However, after the first invasion step, this is exactly what the system looks like, and we want to be able to iterate our procedure. The definition roughly implies that all macroscopic types are close to their coexistence equilibrium (within an attractive domain) and all microscopic types $y$ are of order $\mu^{\lambda_y }$ as $\mu\to0$. The types that are not part of the resident types but of order $\mu^0$ are assumed to be unfit and of small enough size. This ensures that they do not "trigger" the stopping time that marks the beginning of the next Lotka-Volterra phase, i.e.\ the time when the first fit mutant reaches a macroscopic level. This stopping time is defined in (3.1)

Let $\x^0\subset\Hn$ be the initial set of coexisting types, i.e.\ $\xi^\mu_0\in\text{IC}(\x^0,\eta,\bar{c})$, and set $T_0:=0$. During a time of order 1, each type $y\in\Hn$ grows to a size of order $\mu^{\rho^0_y}$, where
\begin{align}
\rho^0_y:=\min_{z\in\Hn}[\lambda_z+|z-y|],
\end{align}
due to incoming mutants from other types.
This can be argued as folllows. The population of type $y$ collects incoming mutants from all other types $z$ of order $\xi^\mu_0(z)\mu^{|z-y|}$. These influences are summed up but in the limit of $\mu\to0$, the asymptotically largest summand, i.e.\ the smallest exponent of $\mu$, dominates all other terms.

Assume now that, after the $(i-1)^\text{st}$ invasion, at time $T_{i-1}\ln 1/\mu$, we have coexisting resident types $\x^{i-1}$ and all types $y\in\Hn$ have population size of order $\mu^{\rho^{i-1}}_y$, where $\rho^{i-1}_y=\min_{z\in\Hn}[\rho^{i-1}_z+|z-y|]$ is satisfied. During a time of order $\ln1/\mu$, microscopic types grow until the first type reaches a population size of order $1$. The population sizes during growth can be approximated as
\begin{align}
\xi^\mu_{t\ln\frac{1}{\mu}}(y)\approx\mu^{\min_{z\in\Hn}[\rho^{i-1}_z+|z-y|-(t-T_{i-1})f_{z,\x^{i-1}}]}.
\end{align}
This is a little tricky and takes into account that there are three possible sources that could  dominate the growth of type $y$:
First, the population at $y$ could just grow at its own exponential growth rate $f_{y,\x^{i-1}}$. This gives $\mu^{\rho^{i-1}_y -(t-T_{i-1})f_{y,\x^{i-1}}}$. 
Second, it could come from mutants from the large populations in $x\in\x^{i-1}$. This gives
$\mu^{|x-y|}$ since $x$ has to mutate $|x-y|$ times to reach $y$.
Finally, it could come from the mutants that have grown at any other site $z$ over the last period. This gives $\mu^{\rho^{i-1}_z+ |z-y|-(t-T_{i-1})f_{z,\x^{i-1}}}$.

Since mutants from another type can never increase the population size past $\mu^1$, the first microscopic type $y$ to reach a size of order 1 must have grown at its own rate $f_{y,\x^{i-1}}>0$. The time to reach this macroscopic size (after the last invasion) is of order $(\rho^{i-1}_y/f_{y,\x^{i-1}})\ln1/\mu$. 

Summarising these thoughts, we inductively define
\begin{align}
y^i_*&:=\argmin_{\substack{y\in\Hn:\\f_{y,\x^{i-1}}>0}}\frac{\rho^{i-1}_y}{f_{y,\x^{i-1}}},\label{unqmin}
\end{align}
the $i^\text{th}$ invading type (if the minimiser is unique),
\begin{align}\label{defTi}
T_i&:=T_{i-1}+\min_{\substack{y\in\Hn:\\f_{y,\x^{i-1}}>0}}\frac{\rho^{i-1}_y}{f_{y,\x^{i-1}}},
\end{align}
the time of the $i^\text{th}$ invasion on the time scale $\ln1/\mu$, and
\begin{align}
\rho^i_y&:=\min_{z\in\Hn}[\rho^{i-1}_z+|z-y|-(T_i-T_{i-1})f_{z,\x^{i-1}}]
\end{align}
the $\mu$-exponent of the population size of type $y$ at the time of the $i^\text{th}$ invasion.
If there is no $y\in\Hn$ such that $f_{y,\x^{i-1}}>0$, we set $T_i:=\infty$.

At time $T_i\ln1/\mu$, the types $y^i_*$ and $\x^{i-1}$ re-equilibrate according to the mutation-free Lotka-Volterra dynamics. If it is unique, we denote the support of the new equilibrium, i.e.\ the new coexisting resident types, by $\x^i$.

\begin{remark}
The results in this paper can easily be generalised to finite, possibly directed graphs as type spaces, where (directed) edges mark the possibility of mutation. In these cases the Hamming distance on the hypercube (e.g.\ $|z-y|$ in (2.22)) is replaced by a ``directed'' distance, corresponding to lengths of directed paths (e.g.\ by the length of the shortest path from $z$ to $y$). Note that this directed distance is not a distance in the classical sense since it might not be symmetric. For ease of notation and due to the nice applicability to genetic sequences, we stick with the hypercube in this paper.
\end{remark}


With the above notation, we can now characterise the limiting process as follows.

\begin{theorem}\label{MainThm}
Consider the system of differential equations (\ref{DE}) and let $\xi^\mu_0\in\text{IC}(\x^0,\eta,\bar{c})$, for $\eta$ small enough. Assume (A) and (B$_{\x^{i-1}\cup y^i_*}$), for every $1\leq i<I$, where we set $I:=i$ for the smallest $i\in\N$ where either
\begin{itemize}
\item[(a)] the minimiser in (\ref{unqmin}) is not unique, or
\item[(b)] there is a $y\in(\x^{i-1}\cup y^i_*)\backslash\x^i$ such that $f_{y,\x^i}\geq0$,
\end{itemize}
and $I:=\infty$ if none of these occur. In the latter case, we set $T_\infty:=\infty$.

Then, for every $t\in[0,T_I)\backslash\{T_i,0\leq i\leq I\}$,
\begin{align}
\lim_{\mu\to0}\xi^\mu_{t\ln\frac{1}{\mu}}=\sum_{i=0}^{I-1} \1_{T_i\leq t< T_{i+1}}\sum_{x\in\x^i}\delta_{x}\bar{\xi}_{\x^i}(x).
\end{align}
\end{theorem}

\begin{remark}
(i) Note that (B$_{\x^{i-1}\cup y^i_*}$) implies (B$_{\x^{i-1}}$) and (B$_{\x^i}$), with the same constants $\theta_x$.

(ii) Case (a) is very unlikely if the parameters of the model are chosen in a random fashion since it requires a very particular equality. Case (b) guarantees that we terminate the procedure as soon as the conditions of $\text{IC}(\x^i,\eta,\bar{c})$ are not satisfied after the $i^\text{th}$ invasion. For every \linebreak $y\in(\x^{i-1}\cup y^i_*)\backslash\x^i$, $f_{y,\x^i}\leq 0$ is ensured (going through the proof of Theorem \ref{LVThm}), so the only problem can arise from equality.

(iii) Note that the theorem implies that, in the case of $T_i=\infty$, even if there was a mutant type $y\in\Hn\backslash\x^{i-1}$ such that $f_{y,\x^{i-1}}=0$, it would not be able to invade the resident population. 
\end{remark}

The proof of this result is given in Section 3 and 4.


\subsection{Convergence in the case of equal competition}

In the context of adaptive walks and flights, the fitness landscape on the type space is possibly random, but usually fixed over time. For a monomorphic resident type, the current fitness of any type, corresponding to its invasion fitness, is determined by the difference between its individual fitness and the fitness of the resident type.

As a special case of our model, we consider equal competition between all types on the hypercube. In this case, one can  simplify the description of the limit process and derive some interesting properties.

We introduce the additional assumption
\begin{itemize}
\item[\textbf{(C)}] For every $x,y\in\Hn$, $\alpha(x,y)\equiv\alpha>0$.
\end{itemize}
This leads to a couple of nice properties of the invasion fitness $f_{y,\{x\}}$. As in the adaptive walks framework, we obtain
\begin{align}
f_{y,\{x\}}=r(y)-\alpha(y,x)\bar{\xi}_{\{x\}}(x)=r(y)-r(x),
\end{align}
which yields
\begin{align}
f_{y,\{x\}}=-f_{x,\{y\}}~\text{ and }~f_{z,\{y\}}+f_{y,\{x\}}=f_{z,\{x\}}.\label{trans}
\end{align}

As a consequence, there is some kind of transitivity of invasion fitness. A type $z$ that is unfit relative to some other type $y$, i.e.\ $f_{z,\{y\}}<0$, is unfit relative to all types that are fitter than $y$. This ensures that types which are once suppressed by resident types stay at a microscopic level forever. In particular, case (b) in Theorem \ref{MainThm} is automatically excluded by assumption (C).

As before, we terminate the procedure as soon as case (a) in Theorem \ref{MainThm} occurs to ensure that there is always a unique mutant that reaches the threshold of order 1 first after an invasion. Starting out with only a single type at its equilibrium size, 
i.e.\ $\x^0=\{x^0\}$, this also implies that we avoid coexistence and always maintain a monomorphic resident population. This is due to the fact that an invading type has to have higher rate $r$ than the current resident type, which prevents a polymorphic Lotka-Volterra equilibrium.

Assumption (B$_\x$) can no longer be satisfied for constant $\alpha$, as soon as $|\x|\geq 2$. However, it is no longer needed since the resident types are monomorphic, i.e.\ $|\x^i|=1$, and have a lower rate $r$ than the invading types, which implies a unique stable equilibrium of the Lotka-Volterra sytem involving $\x^i\cup y^i_*$. We comment on this in more detail in Section 5, where we adapt the proof of Theorem \ref{MainThm} to this situation.

In the case of a monomorphic resident population $\x=\{x\}$, we use the shorthand notation $\bar{\xi}_x:=\bar{\xi}_{\{x\}}$, $f_{y,x}:=f_{y,\{x\}}$. For types $x^i$, $x^j$, we write $f_{i,j}:=f_{x^i,x^j}$.

The limiting jump process can now be described in a simple way.

\begin{theorem}\label{EqComp}
Consider the system of differential equations (\ref{DE}) and let $\xi^\mu_0\in\text{IC}(\{x^0\},\eta,\bar{c})$ such that $\lambda_y\geq |y-x^0|$, for all $y\in\Hn$, and $\eta$ small enough. Assume (A) and (C) and set $I:=i$ for the smallest $i\in\N$ such that the minimiser in (\ref{unqmin}) is not unique, and $I:=\infty$ if this does not occur. In the latter case, we set $T_\infty:=\infty$.

Then, for every $t\in[0,T_I)\backslash\{T_i,0\leq i\leq I\}$,
\begin{align}
\lim_{\mu\to0}\xi^\mu_{t\ln\frac{1}{\mu}}=\sum_{i=0}^{I-1} \1_{T_i\leq t< T_{i+1}}\delta_{x^i}\bar{\xi}_{x^i}(x^i).
\end{align}

Moreover, the following identities hold:
\begin{align}
x^i&=\argmin_{y\in\Hn: f_{y,x^{i-1}}>0}\frac{|y-x^0|-|x^{i-1}-x^0|}{f_{y,x^{i-1}}},\label{argmin}\\
T_i&=\frac{|x^i-x^0|-|x^{i-1}-x^0|}{f_{i,i-1}}.
\end{align}
\end{theorem}

\begin{remark}
(i) $\lambda_y\geq|y-x^0|$ ensures that the initial population size of all microscopic types is not larger than what they gain due to incoming mutants from $x^0$ within a time of order 1. This is neccessary to obtain the identities for $x^i$ and $T_i$.

(ii) Uniqueness of the minimiser in (\ref{unqmin}) is equivalent to uniqueness of the minimiser in (\ref{argmin}) and hence $x^i$ is well-defined.

(iii) In the case where $I=\infty$, the jump process in Theorem \ref{EqComp} continues as long as there is a type with higher individual fitness, i.e.\ higher rate $r$. As a result, it can cross arbitrarily large valleys in the fitness landscape (defined by $r$) and eventually reaches a global fitness maximum, where it remains. Note that this global maximum does not have to be unique. The jump process reaches the maximum that is closest to $x^0$ in $\Hn$, which is unique if $I=\infty$, and equally fit types cannot invade as mentioned in Remark 5(iii).
With these properties, the jump process resembles an adaptive flight. However, it does not quite fit into that framework since it is not only jumping to local fitness maxima.

(iv) Every invasion step increases the distance on $\Hn$ between the resident type and $x_0$. This can be seen inductively as follows. Consider the $(i+1)^\text{st}$ invasion. $x^i$ was a minimiser of $(|y-x^0|-|x^{i-1}-x^0|)/f_{y,x^{i-1}}$. If now $y$ satisfies $f_{y,x^i}>0$, then
\begin{align}
\frac{|y-x^0|-|x^{i-1}-x^0|}{f_{y,x^{i-1}}}\geq\frac{|x^i-x^0|-|x^{i-1}-x^0|}{f_{i,i-1}},
\end{align}
and since $f_{y,x^{i-1}}=f_{y,x^i}+f_{i,i-1}>f_{i,i-1}$ and $|x^i-x^0|>|x^{i-1}-x^0|$ (by assumption), \linebreak $|y-x^0|-|x^{i-1}-x^0|>|x^i-x^0|-|x^{i-1}-x^0|$, and hence $|y-x^0|>|x^i-x^0|$.

The proof of Theorem \ref{EqComp} is found in Section 5.
\end{remark}


\subsection{Derivation from the individual-based stochastic model in the large population limit}

The deterministic system, that is studied above, can be obtained as the large population limit of an individual-based Markov process. At time $t$, we consider a population of finite size $N(t)\in\N$. Each living individual is represented by its type $x_1(t),...,x_{N(t)}(t)\in\Hn$ and the state of the population is described by the finite point measure
\begin{align}
\nu^\mu_t=\sum_{i=1}^{N(t)} \delta_{x_i(t)}.
\end{align}
$\nu^\mu_t(x)$ describes the number of induviduals of type $x\in\Hn$ at time $t$. The dynamics of the Markov process are determined by the same parameters $b$, $d$, $\alpha$, $\mu$, and $m$ as for the deterministic system $\xi^\mu_t$.

To let the size of the population tend to infinity, we introduce the \textit{carrying capacity} of the environment, denoted by $K\in\N$. This can for example be interpreted as the amount of available space or resources. As $K$ increases, the competitive pressure between individuals decreases and we scale $\alpha_K(x,y)\equiv\frac{\alpha(x,y)}{K}$. This leads to an equilibrium population size of order $K$. To derive a finite limit for large populations, i.e.\ as $K\to\infty$, we consider the rescaled measure
\begin{align}
\nu^{\mu,K}_t:=\frac{\nu_t}{K}.
\end{align}

This measure-valued Markov process can be constructed similar to \cite[Ch 2]{FoMe04}, with infinitesimal generator
\begin{align}
\mathcal{L}^K\phi(\nu)=&\sum_{x\in\Hn}K\nu(x)\left(\phi\left(\nu+\frac{\delta_x}{K}\right)-\phi(\nu)\right)b(x)(1-\mu)\notag\\
&+\sum_{x\in\Hn}K\nu(x)\sum_{y\sim x}
\left(\phi\left(\nu+\frac{\delta_y}{K}\right)-\phi(\nu)\right)b(x)\mu m(x,y)\notag\\
&+\sum_{x\in\Hn}K\nu(x)\left(\phi\left(\nu-\frac{\delta_x}{K}\right)-\phi(\nu)\right)\left(d(x)+\sum_{y\in\Hn}\frac{\alpha(x,y)}{K}K\nu(y)\right),
\end{align}
where $\nu\in\MM(\Hn)$ is a non-negative measure on $\Hn$ and $\phi$ a measurable bounded function from $\MM(\Hn)$ to $\R$.

Ethier and Kurtz have shown convergence of this process to $\xi^\mu$ as $K$ tends to infinity.

\begin{theorem}[\cite{EtKu86}, Chap.11, Thm.2.1]
Assume that the initial conditions converge almost surely to a deterministic limit, i.e. $\nu^{\mu,K}_0\to\xi^\mu_0$, as $K\to\infty$. Then, for every $T\geq 0$, $(\nu^{\mu,k}_t)_{0\leq t\leq T}$ almost surely converges uniformly to the deterministic process $(\xi^\mu_t)_{0\leq t\leq T}$, which is the unique solution to the system of differential equations \ref{DE} with initial condition $\xi^\mu_0$.
\end{theorem}

\subsection{Convergence for a limited radius of mutation}
The limiting process in Theorem \ref{EqComp} already looks similar to the greedy adaptive walk of \cite{NoKr15}, mentioned in the introduction. It is a monomorphic jump process on the type space that always jumps to types of higher individual fitness $r$. However, it can take larger steps than just to neighbouring types and we have seen that the initial type $x^0$ plays an important role in determining the jump chain. This is due to the fact that, already after an arbitrarily small time, mutation has induced a positive population size for every possible type. These mutant populations have size of order $\mu$ to the power of the distance to $x^0$ on $\Hn$. The next invading type is then found balancing low initial $\mu$-power and high (invasion) fitness.

In all our previous considerations, arbitrarily small populations were able to reproduce and foster mutants, which can lead to population sizes as small as $\mu^n$. This might not always fit reality well.

If we consider the stochastic model, introduced in the previous subsection, and allow for the mutation probability $\mu_K$ to decrease as $K$ increases, we can study the simultaneous limit of large populations and rare mutations. To be able to reproduce within a time of order 1 in a population of size $\mu_K^n$ implies that
\begin{align}
\lim_{K\to\infty}\mu_K^n\cdot K\geq 1,
\end{align}
or equivalently
\begin{align}
\lim_{K\to\infty}\frac{\mu_K}{K^{-\frac{1}{n}}}\geq 1.
\end{align}
In this case, we would recover the deterministic system (\ref{DE}) in the limit of $K\to\infty$.

If now $\mu_K$ was of order $K^{-\frac{1}{\ell}}$ for some $\ell<n$, this implies that populations with a size of order $\mu_K^\lambda$, for $\lambda>\ell$ are vanishing as $K\to\infty$ and hence cannot reproduce. If we consider a monomorphic resident type $x$, it spreads mutants $y$ of population size $\mu_K^{|y-x|}$. This means that it can initially only foster mutant populations in a radius of $\ell$.

This regime has already been studied by \cite{BoCoSm18} and Champagnat, M\'el\'eard, and Tran (preprint 2019). It is shown that, on the type space $\N$ (with neighbours having difference exactly $1$) and on the usual time scale of $\ln 1/\mu_K$, a fitness valley of width $\leq \ell$, but no further, can be crossed. However, crossing a wider valley is possible on a faster diverging time scale.


In the following, we want to mimic this parameter regime of the stochastic system in our determnistic model. To do so, we introduce a cut-off that freezes the dynamics of a population below the threshold of $\bar{\xi}\mu^\ell$, where $\bar{\xi}:=\min\{\bar{\xi}_x(x)/2: x\in\Hn,r(x)>0\}>0$ is chosen such that every resident type will eventually surpass this value (which is relevant in the case $\ell=1$). The new system of differential equations then reads as follows.
\begin{align}\label{DE'}
\tfrac{d}{dt}\xi^\mu_t(x)=&\ \left[b(x)\1_{\xi^\mu_t(x)\geq\bar{\xi}\mu^\ell}-d(x)-\sum_{y\in\Hn}\alpha(x,y)\xi^\mu_t(y)\1_{\xi^\mu_t(y)\geq\bar{\xi}\mu^\ell}\right]\xi^\mu_t(x)\notag\\
&\ +\mu\sum_{y\sim x}\xi^\mu_t(y)\1_{\xi^\mu_t(y)\geq\bar{\xi}\mu^\ell}b(y)m(y,x)-\mu\xi^\mu_t(x)\1_{\xi^\mu_t(x)\geq\bar{\xi}\mu^\ell}b(x).
\end{align}

\begin{remark}
Reproduction (clonal and non-clonal) is stopped for types below the threshold of $\bar{\xi}\mu^\ell$. As a result, those types are in a kind of dormant state and can only grow due to the mutational influence of other, larger types. It does not affect the system that these dormant types remain at a low level since they do not influence the dynamics of other types and only become active again if they gain a larger amount due to incoming mutants.

The death rate of populations below $\bar{\xi}\mu^\ell$ is not set to zero. This is neccessary to actually drop below the threshold if a population declines due to negative fitness. Otherwise, the population would remain at exactly $\bar{\xi}\mu^\ell$ and could immediately start growing again when its fitness becomes positive due to a change of resident types. This is however not what we want to achieve since populations that drop to the threshold are supposed to go extinct and only reappear due to incoming new mutants.

We cannot simply set the population size of a type to zero below the threshold. In that case, a zero-type would never become active since every gain due to mutation would immediately be cancelled.
\end{remark}

As mentioned above, for $\ell\geq n$, we just recover the original scenario of Theorem \ref{MainThm}. This is due to the fact that, as long as there is at least one macroscopic type, every other type has population size of at least $\mu^n$, due to mutants from this macroscopic type.

For $\ell=1$, if we keep assumption (C) of constant competition and a monomorphic initial condition, we obtain the greedy adaptive walk of \cite{NoKr15}, where the process always jumps to the fittest direct neighbour of the current resident type. We re-define
\begin{align}
x^i&:=\argmax_{y\sim x^{i-1}}r(y),\label{unqmax}\\
T_i&:=T_{i-1}+\frac{1}{f_{i,i-1}},
\end{align}
and set $T_i:=\infty$, as soon as there exists no $y\sim x^{i-1}$ such that $r(y)>r(x^{i-1})$.

The convergence can be stated as follows.

\begin{theorem}\label{AW}
Consider the system of differential equations (\ref{DE'}) for $\ell=1$ and let $\xi^\mu_0\in\text{IC}(\{x^0\},\eta,\bar{c})$ such that $\lambda_y\geq 1$, for all $y\sim x^0$, $\lambda_y\geq 2$, for $|y-x^0|\geq2$, and $\eta$ small enough. Assume (A) and (C) and set $I:=i$ for the smallest $i\in\N$ such that the maximiser in (\ref{unqmax}) is not unique, and $I:=\infty$ if this does not occur. In the latter case, we set $T_\infty:=\infty$.

Then, for every $t\in[0,T_I)\backslash\{T_i,0\leq i\leq I\}$,
\begin{align}
\lim_{\mu\to0}\xi^\mu_{t\ln\frac{1}{\mu}}=\sum_{i=0}^{I-1} \1_{T_i\leq t< T_{i+1}}\delta_{x^i}\bar{\xi}_{x^i}(x^i).
\end{align}
\end{theorem}

\begin{remark}
(i) $\lambda_y\geq 2$, for $|y-x^0|\geq2$, ensures that no microscopic type has a larger initial population that what it gains due to the first incoming mutants from other types.

(ii) The adaptive walk in Theorem \ref{AW} stops as soon as it reaches a local maximum of the individual fitness $r$ since only direct neighbours of the resident type can be reached. Local maxima do not need to be strict. However, as in the previous cases, mutants with invasion fitness $0$ cannot invade the resident population.

(iii) It is no longer the case that every step increases the distance to $x^0$. The walk could return to a type close to $x^0$, which just could not be reached before because one had to go around a valley in the fitness landscape defined by $r$.
\end{remark}

In Section 6, we discuss the proof of Theorem \ref{AW}, as well as the intermediate cases of $1<\ell<n$.


\subsection{Structure of the proofs}

The general strategy of the proofs of all three theorems is to split the analysis of the evolution into two parts. First, the microscopic mutants grow in the presence of the coexisting resident types until one of them reaches a macroscopic population size of order $1$, i.e.\ that does not vanish as $\mu\to0$. Second, this macroscopic mutant and the resident types attain a new equilibrium according to the Lotka-Volterra dynamics. The two phases are visualised in Figure 2, found in Section 4, prior to the proof of Theorem \ref{MainThm}.

The first phase is studied in detail in Section 3. Theorem \ref{Thmexp} gives upper and lower bounds for the exponential growth of the non-resident types. The growth can be due to a type's own (invasion) fitness or due to mutants from a growing neighbour. To get the correct approximation, the influences of all existing types have to be summed up. Meanwhile, the resident types stay close to their equilibrium. Corollary \ref{Corexp} considers the $\ln 1/\mu$-time scale and derives an approximation for the first time that a mutant reaches the macroscopic threshold.

After the threshold is reached, for the second phase, we can apply Theorem \ref{LVThm} to the Lotka-Volterra system involving the macroscopic mutant type and the resident types to derive the convergence to a new equilibrium state. This is possible since we now have a non-negative initial condition that does not vanish as $\mu\to0$.

In Section 4, this theorem is combined with Theorem \ref{Thmexp}, or rather Corollary \ref{Corexp}, to analyse the full evolution of our system (\ref{DE}). First, in Lemma \ref{cont}, the dependence of solutions on the initial condition and the size of $\mu$ is studied to be able to approximate the full system by the Lotka-Volterra system only involving the macroscopic types. Second, in Lemma \ref{conttau}, continuity of the duration of the Lotka-Volterra phase in the initial condition is shown. From this, a uniform bound on the time to reach the initial conditions of Theorem \ref{Thmexp} again is derived.
All of this is then combined to show the convergence in Theorem \ref{MainThm}, one invasion step at a time, and recursively describe the limiting process. 

To prove Theorem \ref{EqComp}, only slight changes have to be made. Since assumption (B$_\x$) is not satisfied, Theorem \ref{LVThm} can no longer be applied directly. However, the assumption is mainly needed to show uniqueness of the limiting equilibrium, which is, in this case, already implied by the structure of the individual fitness landscape. The rest of the proof, found in Section 5, is then devoted to simplifying the expressions for $y^i_*$, $T_i$, and $\rho^i_y$.

In Section 6, Theorem \ref{AW} is proved. Here, the bounds from Theorem \ref{Thmexp} have to be revised. The rest of the argument follows the previous proofs.


\section{Invasion Analysis}

In this section, we prove an exponential approximation of the growth of the non-resident subpopulations until the first type reaches a macroscopic threshold of order 1. We choose this threshold to be at $\eta>0$, independent of $\mu$, and pick $\eta$ small enough for our purposes in the end.
\begin{definition}
For a resident population of $\x\subset\Hn$, the time when the first mutant type reaches $\eta>0$ is defined as
\begin{align}
\tilde{T}^\mu_\eta&:=\inf\{s\geq 0:\exists\ y\in\Hn\backslash\x:\xi^\mu_s(y)>\eta\}.
\end{align}
To consider the evolutionary time scale $\ln1/\mu$, we define $T^\mu_\eta$ through $\tilde{T}^\mu_\eta=T^\mu_\eta\ln1/\mu$.
\end{definition}

We can now state the first result that describes the evolution of the system until $\tilde{T}^\mu_\eta$.
\begin{theorem}\label{Thmexp}
Consider the system of differential equations (\ref{DE}) and assume (A). Then there exist $\tilde{\eta}>0$ and $0<\bar{c}\leq\bar{C}$, uniform in all $\x\subset\Hn$ for which $\text{LVE}_+(\x)=\{\bar{\xi}_\x\}$ and (B$_\x$) is satisfied, such that for $\eta\leq\tilde{\eta}$ and $\mu<\eta$ the following holds:\\
If $\xi^\mu_0\in\text{IC}(\x,\eta,\bar{c})$, then, for every $0<t_0\leq t<\tilde{T}^\mu_\eta$ and every $y\in\Hn$,
\begin{align}
\check{c}\sum_{z\in\Hn}\ee^{t(f_{z,\x}-\eta\check{C})}\mu^{\rho_z+|z-y|}\leq\xi^\mu_t(y)\leq \hat{c}\sum_{z\in\Hn}\ee^{t(f_{z,\x}+\eta\hat{C})}\mu^{\rho_z+|z-y|}(1+t)^m, \label{preexp}
\end{align}
where $\rho_y:=\min_{z\in\Hn}(\lambda_z+|z-y|)$, $m\in\N$, and $0<\check{c},\check{C},\hat{c},\hat{C}<\infty$ are independent of $\mu$ and $\eta$ (but dependent on $t_0$).\\
Moreover, for all $x\in\x$,
\begin{align}
\xi^\mu_t(x)\in[\bar{\xi}_\x(x)-\eta\bar{C},\bar{\xi}_\x(x)+\eta\bar{C}].\label{equieta}
\end{align}
\end{theorem}

As a Corollary, we estimate the growth of the different subpopulations on the time scale $\ln 1/\mu$ and derive the asymptotics of $T^\mu_\eta$ as $\mu\to 0$.
\begin{corollary}\label{Corexp}
Under the same assumptions as in Theorem \ref{Thmexp} and with the same constants, we obtain that, for every $y\in\Hn$ and every $t_0\leq t\ln1/\mu\leq \tilde{T}^\mu_\eta$,
\begin{align}
\check{c}\mu^{\min_{z\in\Hn}[\rho_z+|z-y|-t(f_{z,\x}-\eta\check{C})]}&\leq\xi^\mu_{t\ln\frac{1}{\mu}}(y)\leq 2^n\hat{c}\mu^{\min_{z\in\Hn}[\rho_z+|z-y|-t(f_{z,\x}+\eta\hat{C})]}\left(1+t\ln\frac{1}{\mu}\right)^m.\label{exp}
\end{align}
Moreover, as long as there is a $y\in\Hn$ for which $f_{y,\x}>0$, there is an $\bar{\eta}\leq\tilde{\eta}$ such that for every $\eta\leq\bar{\eta}$
\begin{align}
\min_{\substack{y\in\Hn\\ \lambda_y>0}}\min_{\substack{z\in\Hn\\f_{z,\x}>0}}\frac{\rho_z+|z-y|}{f_{z,\x}+\eta\hat{C}}\leq\liminf_{\mu\to0}T^\mu_\eta\leq\limsup_{\mu\to0}T^\mu_\eta\leq\min_{\substack{y\in\Hn\\ \lambda_y>0}}\min_{\substack{z\in\Hn\\f_{z,\x}>0}}\frac{\rho_z+|z-y|}{f_{z,\x}-\eta\check{C}}.\label{stappr}
\end{align}
\end{corollary}

\begin{proof}[Theorem \ref{Thmexp}]
The proof consists of two steps. We only derive the existence of $\tilde{\eta}$ for a specific set $\x$. To get a uniform parameter, we just have to minimise over the finite set of all such sets $\x$.

First, we show that (\ref{equieta}) holds up to time $\tilde{T}^\mu_\eta$.
Second, we inductively prove the upper bound in (\ref{preexp}). The lower bound can derived analogously.
%

\noindent\textit{Step 1:} $\xi^\mu_t(x)\in[\bar{\xi}_\x(x)-\eta\bar{C},\bar{\xi}_\x(x)+\eta\bar{C}]$.\\
To prove our first claim, we analyse the distance of $\left.\xi^\mu_t\right|_\x:=(\xi^\mu_t(x))_{x\in\x}$ from $\bar{\xi}_\x$ with respect to the norm $\norm{\cdot}_\x$, defined in (\ref{defnorm}). We prove that, in an annulus with respect to the norm $\norm{\cdot}_\x$, this distance declines. Hence, starting inside the annulus, $\left.\xi^\mu_t\right|_\x$ will remain there. This argument is depicted in Figure 1.

To approximate
\begin{align}
\frac{d}{dt}\frac{\norm{\left.\xi^\mu_t\right|_\x-\bar{\xi}_\x}_\x^2}{2}=\left\langle\left.\xi^\mu_t\right|_\x-\bar{\xi}_\x,\frac{d}{dt}(\left.\xi^\mu_t\right|_\x-\bar{\xi}_\x)\right\rangle_\x
\end{align}
from above, we split the right hand side of (\ref{DE}) into two parts.

We define $F,V:\mathcal{M}(\Hn)\to\R^\x$,
\begin{align}
F_x(\xi)=\left[r(x)-\sum_{y\in\x}\alpha(x,y)\xi(y)\right]\xi(x),\quad x\in\x,
\end{align}
the Lotka-Volterra part, and 
\begin{align}
V_x(\xi)=-\sum_{y\in\Hn\backslash\x}\alpha(x,y)\xi(y)\xi(x)
+\mu\sum_{y\sim x}b(y)m(y,x)\xi(y)-\mu b(x)\xi(x),\quad x\in\x,
\end{align}
the error part of the differential equation.

With this,
\begin{align}
\left\langle\left.\xi^\mu_t\right|_\x-\bar{\xi}_\x,\frac{d}{dt}(\left.\xi^\mu_t\right|_\x-\bar{\xi}_\x)\right\rangle_\x
=\langle\left.\xi^\mu_t\right|_\x-\bar{\xi}_\x,F(\xi^\mu_t)\rangle_\x
+\langle\left.\xi^\mu_t\right|_\x-\bar{\xi}_\x,V(\xi^\mu_t)\rangle_\x.
\end{align}

We first approximate the norm of the error part, using that $|\xi^\mu_t(y)|\leq\eta$ for $y\in\Hn\backslash\x$. In addition, we assume that, for every $x\in\x$, $\xi^\mu_t(x)\geq\eta$. We choose $\eta$ such that this is always implied by (\ref{equieta}) at the end of Step 1.

We estimate
\begin{align}
|V_x(\xi^\mu_t)|\leq \eta2^n\max_{x\in\x,y\in\Hn}\alpha(x,y)|\xi^\mu_t(x)|+\mu n\max_{y\in\Hn}b(y)\max_{y\in\Hn}|\xi^\mu_t(y)|+\mu\max_{y\in\Hn}b(y)|\xi^\mu_t(x)|
\end{align}
and hence, using that $\max_{y\in\Hn}|\xi^\mu_t(y)|\leq\norm{\left.\xi^\mu_t\right|_\x}\leq c_\x^{-1}\norm{\left.\xi^\mu_t\right|_\x}$,
\begin{align}\label{errorest}
\norm{V(\xi^\mu_t)}_\x\leq&\ \eta2^n\max_{x\in\x,y\in\Hn}\alpha(x,y)\norm{\left.\xi^\mu_t\right|_\x}_\x\notag\\
&\ +\mu\max_{y\in\Hn}b(y)\left(n\sqrt{|\x|\max_{x\in\x}\frac{\theta_x}{\bar{\xi}_\x(x)}\max_{y\in\Hn}|\xi^\mu_t(y)|^2}+\norm{\left.\xi^\mu_t\right|_\x}_\x\right)\notag\\
\leq&\ \eta2^n\max_{x\in\x,y\in\Hn}\alpha(x,y)\norm{\left.\xi^\mu_t\right|_\x}_\x+\mu\max_{y\in\Hn}b(y)\left(n\sqrt{|\x|} C_\x c_\x^{-1}+1\right)\norm{\left.\xi^\mu_t\right|_\x}_\x\notag\\
\leq&\ \eta\norm{\left.\xi^\mu_t\right|_\x}_\x C,
\end{align}
for some $C<\infty$ independent of $\eta$ and $\mu$.

Next, we approximate the Lotka-Volterra part. To do so, we show that a slight perturbation of the positive definite matrix $(\theta_x\alpha(x,y))_{x,y\in\x}$ is still positive definite. Let $\zeta\in\R^\x$ such that, for $x\in\x$, $|\zeta(x)-1|\leq \tilde{\eps}_\x$. Then
\begin{align}
\sum_{x,y\in\x}\zeta(x)\theta_x\alpha(x,y)u(x)u(y)\notag&=\sum_{x,y\in\x}\theta_x\alpha(x,y)u(x)u(y)+\sum_{x,y\in\x}(\zeta(x)-1)\theta_x\alpha(x,y)u(x)u(y)\notag\\
&\geq \kappa\norm{u}^2-\max_{x\in\x}|\zeta(x)-1|\max_{x,y\in\x}(\theta_x\alpha(x,y))\sum_{x,y\in\x} |u(x)||u(y)|\notag\\
&\geq \norm{u}^2\big[\kappa-\tilde{\eps}_\x|\x|^2\max_{x,y\in\x}\theta_x\alpha(x,y)\big]\geq\frac{\kappa}{2}\norm{u}^2,
\end{align}
as long as $\tilde{\eps}_\x\leq\kappa(2|\x|^2\max_{x,y\in\x}\theta_x\alpha(x,y))^{-1}$.

We now apply this to $\zeta(x)=\xi^\mu_t(x)/\bar{\xi}_\x(x)$. The condition $|\zeta(x)-1|\leq\tilde{\eps}_\x$ is satisfied whenever 
\begin{align}
|\xi^\mu_t(x)-\bar{\xi}_\x(x)|\leq\tilde{\eps}_\x\bar{\xi}_\x(x),
\end{align}
which is the case if
\begin{align}
\norm{\left.\xi^\mu_t\right|_\x-\bar{\xi}_\x}_\x \leq c_\x\tilde{\eps}_\x\min_{x\in\x}\bar{\xi}_\x(x)=:\eps_\x.
\end{align}

Using the fact that $\bar{\xi}_\x$ is an equilibrium of (\ref{equil}) for which $\bar{\xi}_\x(x)>0$ holds for all $x\in\x$, we derive
\begin{align}\label{LVest}
\langle\left.\xi^\mu_t\right|_\x-\bar{\xi}_\x,F(\xi^\mu_t)\rangle_\x 
=&\ \sum_{x\in\x} \frac{\theta_x}{\bar{\xi}_\x(x)}(\xi^\mu_t(x)-\bar{\xi}_\x(x))
 \left[r(x)-\sum_{y\in\x}\alpha(x,y)\xi^\mu_t(y)\right]\xi^\mu_t(x)\notag\\
=&\ -\sum_{x\in\x} \frac{\theta_x}{\bar{\xi}_\x(x)}(\xi^\mu_t(x)-\bar{\xi}_\x(x))
 \left[\sum_{y\in\x}\alpha(x,y)(\xi^\mu_t(y)-\bar{\xi}_\x(y))\right]\xi^\mu_t(x)\notag\\
=&\ -\sum_{x,y\in\x}\frac{\xi^\mu_t(x)}{\bar{\xi}_\x(x)}\theta_x\alpha(x,y)(\xi^\mu_t(x)-\bar{\xi}_\x(x))(\xi^\mu_t(y)-\bar{\xi}_\x(y))\notag\\
\leq&\ -\frac{\kappa}{2}\norm{\left.\xi^\mu_t\right|_\x-\bar{\xi}_\x}^2.
\end{align}

Combining estimates (\ref{errorest}) and (\ref{LVest}), we get
\begin{align}
\frac{d}{dt}\frac{\norm{\left.\xi^\mu_t\right|_\x-\bar{\xi}_\x}_\x ^2}{2}
 =&\ \langle\left.\xi^\mu_t\right|_\x-\bar{\xi}_\x,F(\xi^\mu_t)\rangle_\x 
 +\langle\left.\xi^\mu_t\right|_\x-\bar{\xi}_\x,V(\xi^\mu_t)\rangle_\x \notag\\
\leq&\ -\frac{\kappa}{2}\norm{\left.\xi^\mu_t\right|_\x-\bar{\xi}_\x}^2
 +\norm{\left.\xi^\mu_t\right|_\x-\bar{\xi}_\x}_\x \norm{V(\xi^\mu_t)}_\x \notag\\
\leq&\ -\frac{\kappa}{2}\norm{\left.\xi^\mu_t\right|_\x-\bar{\xi}_\x}^2
 +\norm{\left.\xi^\mu_t\right|_\x-\bar{\xi}_\x}_\x \eta\norm{\left.\xi^\mu_t\right|_\x}_\x C\notag\\
\leq&\ -\norm{\left.\xi^\mu_t\right|_\x-\bar{\xi}_\x}_\x ^2\left(\frac{\kappa}{2C_\x^2}-\eta\frac{C\norm{\left.\xi^\mu_t\right|_\x}_\x }{\norm{\left.\xi^\mu_t\right|_\x-\bar{\xi}_\x}_\x }\right)\notag\\
\leq&\ -\norm{\left.\xi^\mu_t\right|_\x-\bar{\xi}_\x}_\x ^2\frac{\kappa}{4C_\x^2}<0,\label{negderiv}
\end{align}
whenever
\begin{align}
\eps_\x\geq\norm{\left.\xi^\mu_t\right|_\x-\bar{\xi}_\x}_\x\geq\eta C\norm{\left.\xi^\mu_t\right|_\x}_\x\frac{4C_\x^2}{\kappa}\geq\eta C(\norm{\bar{\xi}_\x}_\x-\eps_\x)\frac{4C_\x^2}{\kappa}=:\eta \underline{c}.
\end{align}

Finally, we choose $\tilde{\eta}$ small enough such that $\tilde{\eta}<\eps_\x/\underline{c}$.

Now we can follow the argument that was outlined in the beginning and is supported by Figure 1. As long as $\eta\leq\tilde{\eta}$ and
\begin{align}
\norm{\left.\xi^\mu_0\right|_\x-\bar{\xi}_\x}\leq\eta \underline{c}C_\x^{-1}=:\eta\bar{c}_\x,
\end{align}
we obtain that $\norm{\left.\xi^\mu_0\right|_\x-\bar{\xi}_\x}_\x\leq\eta c$. Because of (\ref{negderiv}), we obtain that $\norm{\left.\xi^\mu_t\right|_\x-\bar{\xi}_\x}_\x\leq\eta \underline{c}$, for every $0\leq t\leq \tilde{T}^\mu_\eta$, and hence
\begin{align}
\norm{\left.\xi^\mu_t\right|_\x-\bar{\xi}_\x}\leq\eta \underline{c}c_\x^{-1}=: \eta \bar{C}_\x.
\end{align}
For the single types, this implies, for every $0\leq t\leq \tilde{T}^\mu_\eta$, that
\begin{align}
\xi^\mu_t(x)\in[\bar{\xi}_\x(x)-\eta \bar{C}_\x,\bar{\xi}_\x(x)+\eta \bar{C}_\x],\quad x\in\x,
\end{align}
whenever
\begin{align}
\xi^\mu_0(x)\in\left[\bar{\xi}_\x(x)-\eta\frac{\bar{c}_\x}{\sqrt{|\x|}},\bar{\xi}_\x(x)+\eta\frac{\bar{c}_\x}{\sqrt{|\x|}}\right],\quad x\in\x.
\end{align}
Setting $\bar{c}:=\min_{\textbf{y}\subset\Hn}\bar{c}_\textbf{y}$ and $\bar{C}:=\max_{\textbf{y}\subset\Hn}\bar{C}_\textbf{y}$, and
choosing $\tilde{\eta}\leq\min_{x\in\x}\bar{\xi}_\x(x)/(2\bar{C}+2)$ to ensure that $\xi^\mu_t(x)>\eta$, for every $x\in\x$, we arrive at the claim.

\begin{figure}[h]
\centering
 \scriptsize
 \begin{tikzpicture}
 \draw[thick] (0,0) ellipse (2cm and 1cm);
 \draw[thick] (0,0) ellipse (3cm and 1.5cm);
 \draw[thick,->] (-2.9,0)--(-2.1,0);
 \draw[thick,->] (0,1.4)--(0,1.1);
 \draw[thick,->] (2.9,0)--(2.1,0);
 \draw[thick,->] (0,-1.4)--(0,-1.1);
 \draw[thick,dashed] (0,0) circle (1);
 \draw[thick,dashed] (0,0) circle (2);
 \fill[black] (0,0) circle (0.05) node[anchor=north]{$\bar{\xi}_\x$};
 \fill[black] (-0.5,0.2) circle (0.05) node[anchor=north]{$\left.\xi^\mu_0\right|_\x$};
 \draw (-0.5,0.2) edge[out=25,in=170,->,thick,dotted] (1.4,0.4);
 \fill[black] (1.4,0.4) circle (0.05) node[anchor=north]{$\left.\xi^\mu_t\right|_\x$};
 \end{tikzpicture}
 \caption{Scheme for the argument in Step 1. Dashed lines indicate balls $B(\bar{\xi}_\x,\eta\bar{c}_\x)$ and $B(\bar{\xi}_\x,\eta\bar{C}_\x)$ with respect to the standard Euclidean norm, while solid lines correspond to balls $B_\x(\bar{\xi}_\x,\eta c)$ and  $B_\x(\bar{\xi}_\x,\eps_\x)$ with respect to the $\norm{\cdot}_\x$ norm.}
 \end{figure}
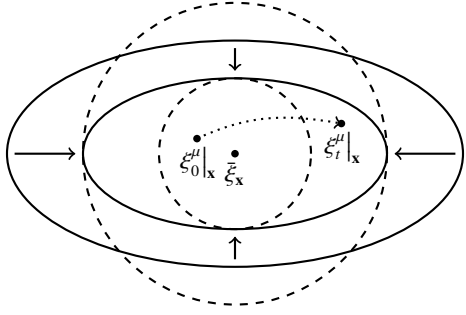

%
\noindent\textit{Step 2:} Inductive exponential bounds.\\
We derive the upper bound for $\xi^\mu_t(y)$ in (\ref{preexp}) in full length. At the end of the proof, we comment on how the same strategy can be adapted to the lower bound. 

To begin, we establish an upper bound on $\tfrac{d}{dt}\xi^\mu_t$.
\begin{align}
\tfrac{d}{dt}\xi^\mu_t(y)&\leq\Big[r(y)-\sum_{x\in\x}\alpha(y,x)(\bar{\xi}_\x(x)-\eta\bar{C} )\Big]\xi^\mu_t(y)+\mu\sum_{z\sim y}\underbrace{b(z)m(z,y)}_{\leq \tilde{C}_y\forall z\sim y}\xi^\mu_t(z)\notag\\
&\leq\Big[r(y)-\sum_{x\in\x}\alpha(y,x)\bar{\xi}_\x(x)+\eta\bar{C}\underbrace{\sum_{x\in\x}\alpha(y,x)}_{=:\hat{C}_y}\Big]\xi^\mu_t(y)+\mu\tilde{C}_y\sum_{z\sim y}\xi^\mu_t(z)\notag\\
&\leq[f_{y,\x}+\eta\hat{C}]\xi^\mu_t(y)+\mu\tilde{C}\sum_{z\sim y}\xi^\mu_t(z),\label{upbnd}
\end{align}
where $\hat{C}:=\max_{y\in\Hn}\hat{C}_y<\infty$ and $\tilde{C}:=\max_{y\in\Hn}\tilde{C}_y<\infty$.

We prove by induction that, for every $m\geq 0$, there exists a constant $C_m<\infty$, independent of $\mu$, $\eta$, and $y$, such that, for every $0\leq t\leq \tilde{T}^\mu_\eta$,
\begin{align}
\xi^\mu_t(y)\leq C_m\Bigg[ \sum_{\substack{z\in\Hn\\|z-y|\leq m}}\ee^{t(f_{z,\x}+\eta\hat{C})}\Big(\mu^{\rho_z+|z-y|}+\frac{1}{\eta}\mu^{m+1}\Big)(1+t)^m+\mu^{m+1}\Bigg].\label{ind3}
\end{align}

For the case $m=0$, we approximate
\begin{align}
\tfrac{d}{dt}\xi^\mu_t(y)\leq[f_{y,\x}+\eta\hat{C}]\xi^\mu_t(y)+\mu\tilde{C}\underbrace{\sum_{z\sim y}\1_{z\in\x}(\bar{\xi}_\x(z)+\eta\bar{C} )+\1_{z\in\Hn\backslash\x}\eta}_{\leq C\text{ uniformly in }y,z},
\end{align}
and hence 
\begin{align}
\xi^\mu_t(y)&\leq \ee^{t(f_{y,\x}+\eta\hat{C})}\xi^\mu_0(y)+\mu\tilde{C}C\int_0^t \ee^{(t-s)(f_{y,\x}+\eta\hat{C})}ds\notag\\
&\leq \ee^{t(f_{y,\x}+\eta\hat{C})}C_y\mu^{\lambda_y}+\mu\tilde{C}C\frac{1}{f_{y,\x}+\eta\hat{C}}(\ee^{t(f_{y,\x}+\eta\hat{C})}-1).
\end{align}
Choose $\tilde{\eta}>0$ small enough such that $f_{y,\x}+\tilde{\eta}\hat{C}< 0$, for every $y\in\Hn$ for which $f_{y,\x}<0$. Then, for $\eta\leq\tilde{\eta}$ and a different constant $C<\infty$, the second summand can be bounded from above by $C\mu$, for $f_{y,\x}<0$, and by $C/\eta\cdot \ee^{t(f_{y,\x}+\eta\hat{C})}\mu$, for $f_{y,\x}\geq0$. $C$ can be chosen independent of $y$, $\mu$, $\eta\leq\tilde{\eta}$, and $0\leq t\leq \tilde{T}^\mu_\eta$. Overall, using $\lambda_y\geq\rho_y$, we get
\begin{align}
\xi^\mu_t(y)\leq \underbrace{((\max_{y\in\Hn}C_y)\lor C)}_{=:C_0<\infty}\Big[\ee^{t(f_{y,\x}+\eta\hat{C})}\Big(\mu^{\rho_y}+\frac{1}{\eta}\mu\Big)+\mu\Big],
\end{align}
which is the desired bound.\\
Assuming that the hypothesis holds for $m-1$ and using (\ref{upbnd}), we approximate
\begin{align}
\tfrac{d}{dt}\xi^\mu_t(y)\leq&\ [f_{y,\x}+\eta\hat{C}]\xi^\mu_t(y)\notag\\
&\ +\mu\tilde{C}\sum_{z\sim y}C_{m-1}\Bigg[\sum_{\substack{u\in\Hn\\|u-z|\leq m-1}}\ee^{t(f_{u,\x}+\eta\hat{C})}\Big(\mu^{\rho_u+|u-z|}+\frac{1}{\eta}\mu^m\Big)(1+t)^{m-1}+\mu^m\Bigg].
\end{align}
Splitting up the second summand, Gronwall's inequality yields
\begin{align}
\xi^\mu_t(y)\leq&\ \ee^{t(f_{y,\x}+\eta\hat{C})}\xi^\mu_0(y)+\tilde{C}C_{m-1}n\mu^{m+1}
  \int_0^t\ee^{(t-s)(f_{y,\x}+\eta\hat{C})}ds\notag\\
&\ +\tilde{C}C_{m-1}
  \sum_{z\sim y}\sum_{\substack{u\in\Hn\\|u-z|\leq m-1}}
  \Big(\mu^{\rho_u+|u-z|+1}+\frac{1}{\eta}\mu^{m+1}\Big)\notag\\
&\ \qquad\cdot\int_0^t(1+s)^{m-1}\ee^{s(f_{u,\x}+\eta\hat{C})}\ee^{(t-s)(f_{y,\x}+\eta\hat{C})}ds  \notag\\
\leq&\ \ee^{t(f_{y,\x}+\eta\hat{C})}C_y\mu^{\lambda_y}+C
  \mu^{m+1}\Big(1+\frac{1}{\eta}\ee^{t(f_{y,\x}+\eta\hat{C})}\Big)\notag\\
&\ +\tilde{C}C_{m-1}
  \sum_{z\sim y}\sum_{\substack{u\in\Hn\\|u-z|\leq m-1}}
  \Big(\mu^{\rho_u+|u-z|+1}+\frac{1}{\eta}\mu^{m+1}\Big)(1+t)^{m-1}\notag\\
&\ \qquad\cdot\int_0^t\ee^{t(f_{y,\x}+\eta\hat{C})}\ee^{s(f_{u,\x}-f_{y,\x})}ds ,\label{upbnd1}
\end{align}
where we bound the first integral just as before in the base case.

We distinguish two cases to approximate the second integral. If $f_{u,\x}\neq f_{y,\x}$, then
\begin{align}
\int_0^t \ee^{t(f_{y,\x}+\eta\check{C})}\ee^{s(f_{u,\x}-f_{y,\x})}ds&=\frac{1}{f_{u,\x}-f_{y,\x}}(\ee^{t(f_{u,\x}+\eta\check{C})}-\ee^{t(f_{y,\x}+\eta\check{C})})\notag\\
&=\frac{1}{|f_{u,\x}-f_{y,\x}|}|\ee^{t(f_{u,\x}+\eta\check{C})}-\ee^{t(f_{y,\x}+\eta\hat{C})}|\notag\\
&\leq C'(\ee^{t(f_{u,\x}+\eta\check{C})}+\ee^{t(f_{y,\x}+\eta\check{C})}),
\end{align}
for some $C'<\infty$ large enough, uniformly in $y$ and $u$.\\
If $f_{u,\x}=f_{y,\x}$, then
\begin{align}
\int_0^t \ee^{t(f_{y,\x}+\eta\check{C})}\ee^{s(f_{u,\x}-f_{y,\x})}ds=t\ee^{t(f_{y,\x}+\eta\check{C})}.
\end{align}
Plugging this back into (\ref{upbnd1}) we get
\begin{align}
\xi^\mu_t(y)\leq&\ee^{t(f_{y,\x}+\eta\hat{C})}C_y\mu^{\lambda_y}+C
  \mu^{m+1}\Big(1+\frac{1}{\eta}\ee^{t(f_{y,\x}+\eta\hat{C})}\Big)\notag\\
&+\tilde{C}C_{m-1} 
  \sum_{z\sim y}\sum_{\substack{u\in\Hn\\|u-z|\leq m-1}}
  \Big(\mu^{\rho_u+|u-z|+1}+\frac{1}{\eta}\mu^{m+1}\Big)(1+t)^{m-1}\notag\\
&\ \qquad\cdot C'(1+t)(\ee^{t(f_{u,\x}+\eta\hat{C})}+\ee^{t(f_{y,\x}+\eta\hat{C})})\notag\\
\leq&\underbrace{(n+o(1))((\max_{y\in\Hn}C_y)\lor C\lor\tilde{C}C_{m-1}C')}_{\leq C_m\text{ for }\mu<\tilde{\eta}}\notag\\
  &\qquad\cdot\Bigg[\sum_{\substack{z\in\Hn\\|z-y|\leq m}}\ee^{t(f_{z,\x}+\eta\hat{C})}\Big(\mu^{\rho_z+|z-y|}+\frac{1}{\eta}\mu^{m+1}\Big)(1+t)^m+\mu^{m+1}\Bigg],
\end{align}
where we used that $\rho_y\leq(\rho_u+|u-z|+1)\land\lambda_y$ for all $z\sim y$ and $|u-z|\leq m-1$, and gathered all the higher $\mu$-powers in the $o(1)$ with respect to the limit $\mu\to 0$. This concludes the proof of (\ref{ind3}).\\
Finally, we can choose $m\geq\max_{y\in\Hn}\max_{z\in\Hn}\rho_z+|z-y|\geq n$ and, since $f_{z,\x}=0$ for all $z\in\x$, we get
\begin{align}
\xi^\mu_t(y)&\leq C_m\Bigg[ \sum_{z\in\Hn}\ee^{t(f_{z,\x}+\eta\hat{C})}\Big(\mu^{\rho_z+|z-y|}+\frac{1}{\eta}\mu^{m+1}\Big)(1+t)^m+\mu^{m+1}\Bigg]\notag\\
&\leq C_m\Bigg[ \sum_{z\in\Hn}\ee^{t(f_{z,\x}+\eta\hat{C})}(\mu^{\rho_z+|z-y|}+\mu^m)(1+t)^m+\sum_{z\in\x}\ee^{t(f_{z,\x}+\eta\hat{C})}\mu^{m+1}\Bigg]\notag\\
&\leq 3C_m\sum_{z\in\Hn}\ee^{t(f_{z,\x}+\eta\hat{C})}\mu^{\rho_z+|z-y|}(1+t)^m.
\end{align}
With $\hat{c}:=3C_m$ and choosing $\tilde{\eta}$ uniform over all subsets $\x\subset\Hn$ of coexisting resident types, this yields the desired upper bound.

The proof of the lower bound is very similar. We approximate, for every $y\in\Hn$,
\begin{align}
\tfrac{d}{dt}\xi^\mu_t(y)\geq[f_{y,\x}-\eta\check{C}]\xi^\mu_t(y)+\mu\tilde{c}\sum_{z\sim y}\xi^\mu_t(z),\label{lowbnd}
\end{align}
and then use the inductive application of Gronwall's inequality twice.

First, to prove that, for an arbitrarily small $t_0>0$, $\xi^\mu_{t_0/2}(y)\geq c_{t_0}\mu^{\rho_y}$, where $c_{t_0}>0$ can be chosen uniformly in $\mu$, $\eta$, and $y$. This corresponds to mutation producing a positive population size for every type within a time of order 1.

Second, we show that, for every $0\leq m\leq n$, there exists a constant $c_m>0$, independent of $\mu$, $\eta$, and $y$, such that, for $(n+m){t_0}/(2n)\leq t\leq \tilde{T}^\mu_\eta$,
\begin{align}
\xi^\mu_t(y)\geq c_m \sum_{\substack{z\in\Hn\\|z-y|\leq m}}\mu^{\rho_z+|z-y|}\ee^{t(f_{z,\x}-\eta\check{C})}.\label{ind2}
\end{align}

Setting $\check{c}:=c_n$ yields the lower bound in (\ref{preexp}), for ${t_0}\leq t\leq \tilde{T}^\mu_\eta$.
\end{proof}


\begin{proof}[Corollary \ref{Corexp}]
The inequalities in (\ref{exp}) follow directly from (\ref{preexp}) by inserting the new time scale. For the lower bound, only the asymptotically largest summand, corresponding to the smallest $\mu$-power, is kept. For the upper bound, every one of the $2^n$ summands is estimated against this largest one.

To prove the second part of the corollary, we first show that, for $\mu$ small enough, the first non-resident type $y$ that reaches the $\eta$-threshold, i.e.\ the type that determines the stopping time $T^\mu_\eta$, satisfies $\lambda_y>0$ and hence $\rho_z+|z-y|>0$, for every $z\in\Hn$.\\
Let $y\in\Hn\backslash\x$ be a non-resident type for which $\lambda_y=0$. This implies $\xi^\mu_0(y)\leq\eta/3$ and $f_{y,\x}<0$. Going back into the proof of (\ref{ind3}) and using that $\tilde{\eta}$ is chosen such that $f_{y,\x}+\tilde{\eta}\hat{C}<0$, this yields
\begin{align}
\xi^\mu_t(y)&\leq \ee^{t(f_{y,\x}+\eta\hat{C})}C_y\mu^{\lambda_y}+\mu\tilde{C}C\frac{1}{f_{y,\x}+\eta\hat{C}}(\ee^{t(f_{y,\x}+\eta\hat{C})}-1)\notag\\
&\leq \ee^{t(f_{y,\x}+\eta\hat{C})}\frac{\eta}{3}+\mu\tilde{C}C\frac{1}{|f_{y,\x}+\eta\hat{C}|}(1-\ee^{t(f_{y,\x}+\eta\hat{C})})\notag\\
&\leq \frac{\eta}{3}+\frac{\mu\tilde{C}C}{|f_{y,\x}+\tilde{\eta}\hat{C}|}\leq\frac{2}{3}\eta,
\end{align}
whenever $\mu\leq\eta|f_{y,\x}+\tilde{\eta}\hat{C}|/3\tilde{C}C$. As a consequence, as $\mu\to0$, $y$ stays strictly below $\eta$ and does not determine $T^\mu_\eta$.

Now we assume that $T^\mu_\eta$ is determined by a non-resident type $y\in\Hn$ for which $\lambda_y>0$, i.e.\ $y$ is the first mutant to reach the $\eta$-threshold. Let $\bar{\eta}\leq\tilde{\eta}\land1\land\check{c}$. Then, assuming that $0<\mu\leq\eta\leq\bar{\eta}$, the lower bound in (\ref{exp}) yields 
\begin{align}
\check{c}\mu^{\min_{z\in\Hn}[\rho_z+|z-y|-T^\mu_\eta(f_{z,\x}-\eta\check{C})]}\leq\xi^\mu_{\tilde{T}^\mu_\eta}(y)=\eta,
\end{align}
and hence
\begin{align}
\ln(\mu)\min_{z\in\Hn}[\rho_z+|z-y|-T^\mu_\eta(f_{z,\x}-\eta\check{C})]\leq\ln\left(\frac{\eta}{\check{c}}\right)\leq 0.
\end{align}
Since $\ln(\mu)<0$, we obtain, for every $z\in\Hn$, that
\begin{align}
\rho_z+|z-y|\geq T^\mu_\eta(f_{z,\x}-\eta\check{C}),
\end{align}
and therefore, if we choose $\bar{\eta}$ small enough such that, for every $\eta\leq\bar{\eta}$ and every $z\in\Hn$ for which $f_{z,\x}>0$, also $f_{z,\x}-\eta\check{C}>0$,
\begin{align}
T^\mu_\eta\leq\min_{\substack{z\in\Hn\\f_{z,\x}>0}}\frac{\rho_z+|z-y|}{f_{z,\x}-\eta\check{C}}.\label{upperT}
\end{align}

To get a lower bound for $T^\mu_\eta$, (\ref{exp}) implies
\begin{align}
\eta=\xi^\mu_{\tilde{T}^\mu_\eta}(y)\leq 2^n\hat{c}\mu^{\min_{z\in\Hn}[\rho_z+|z-y|-T^\mu_\eta(f_{z,\x}+\eta\hat{C})]}\left(1+\tilde{T}^\mu_\eta\right)^m,
\end{align}
which yields
\begin{align}
\ln(\mu)\min_{z\in\Hn}[\rho_z+|z-y|-T^\mu_\eta(f_{z,\x}+\eta\hat{C})]\geq\ln\left(\frac{\eta}{2^n\hat{c}(1+\tilde{T}^\mu_\eta)^m}\right),
\end{align}
and therefore there exists a $z\in\Hn$ such that
\begin{align}
\rho_z+|z-y|\leq T^\mu_\eta(f_{z,\x}+\eta\hat{C})+\frac{\ln\left(\frac{2^n\hat{c}}{\eta}\right)+m\ln(1+\tilde{T}^\mu_\eta)}{\ln\frac{1}{\mu}}.
\end{align}
The second summand on the right hand side is positive and, with (\ref{upperT}), converges to zero as $\mu\to0$. Since the left hand side is positive this implies that $f_{z,\x}+\eta\hat{C}>0$ and by our choice of $\tilde{\eta}$ in the proof of (\ref{ind3}) we obtain $f_{z,\x}\geq0$.\\
Consequently, for every fixed $0<\eta\leq\bar{\eta}$, it follows that
\begin{align}
\liminf_{\mu\to0}T^\mu_\eta\geq\frac{\rho_z+|z-y|}{f_{z,\x}+\eta\check{C}}\geq\min_{\substack{z\in\Hn\\f_{z,\x}\geq0}}\frac{\rho_z+|z-y|}{f_{z,\x}+\eta\check{C}}.
\end{align}

Overall, for every fixed $0<\eta\leq\bar{\eta}$, we obtain
\begin{align}
\min_{\substack{z\in\Hn\\f_{z,\x}\geq0}}\frac{\rho_z+|z-y|}{f_{z,\x}+\eta\check{C}}\leq \liminf_{\mu\to0}T^\mu_\eta\leq\limsup_{\mu\to0}T^\mu_\eta\leq\min_{\substack{z\in\Hn\\f_{z,\x}>0}}\frac{\rho_z+|z-y|}{f_{z,\x}-\eta\check{C}}.
\end{align}
If we now pick $\bar{\eta}$ small enough, both minima are realised by the same $z\in\Hn$ for which $f_{z,\x}>0$, that also minimise
\begin{align}
\min_{\substack{z\in\Hn\\f_{z,\x}>0}}\frac{\rho_z+|z-y|}{f_{z,\x}},
\end{align}
and we can reduce to only considering $z\in\Hn$ such that $f_{z,\x}>0$ in the lower bound.

All the above considerations apply to a single $y$ for which $\lambda_y>0$. Considering all such $y\in\Hn$ we get that asymptotically
\begin{align}
\min_{\substack{y\in\Hn\\ \lambda_y>0}}\min_{\substack{z\in\Hn\\f_{z,\x}>0}}\frac{\rho_z+|z-y|}{f_{z,\x}+\eta\check{C}}\leq \liminf_{\mu\to 0}T^\mu_\eta\leq \limsup_{\mu\to 0} T^\mu_\eta\leq\min_{\substack{y\in\Hn\\ \lambda_y>0}}\min_{\substack{z\in\Hn\\f_{z,\x}>0}}\frac{\rho_z+|z-y|}{f_{z,\x}-\eta\check{C}}. 
\end{align}
For the upper bound, the minimum can be used since, if $T^\mu_\eta$ was larger than this minimum, the minimiser would reach the $\eta$-level before $\tilde{T}^\mu_\eta$, which would be a contradiction.\\
This finishes the proof of the corollary.
\end{proof}


\section{Construction of the Jump Process}

In this section we combine the results of Theorem \ref{Thmexp}, or rather Corollary \ref{Corexp}, and Theorem \ref{LVThm} to derive the convergence of $\xi^\mu$ as $\mu\to0$ to a jump process that moves between Lotka-Volterra equilibria of coexistence. We prove the convergence by an induction over the invasion steps and show that after each invasion the criteria for the initial conditions in Theorem \ref{Thmexp} are again satisfied.

Before we get to the actual proof, we derive two lemmas. The first lemma treats the boundedness of solutions of (\ref{DE}), the continuity in the initial condition, and the perturbation through the mutation rate $\mu$.
\begin{lemma}\label{cont}
Let
\begin{align}
\Omega:=\left\{\xi\in\MM(\Hn):\forall x\in\Hn: \xi(x)\in\left[0,2\frac{|r(x)|}{\alpha(x,x)}\right]\right\}.
\end{align}
There is a $\mu_0>0$ such that, for every $0\leq\mu<\mu_0$, for every $\xi^\mu_0\in\Omega$, and for every $t\geq0$, we obtain $\xi^\mu_t\in\Omega$, where $\xi^\mu_t$ is the solution of (\ref{DE}).

Moreover, there are positive, finite constants $A$, $B$ such that, for every $0\leq\mu_1,\mu_2<\mu_0$, for every $\xi^{\mu_1}_0,\xi^{\mu_2}_0\in\Omega$, and every $t\geq s\geq0$,
\begin{align}
\norm{\xi^{\mu_1}_t-\xi^{\mu_2}_t}\leq \ee^{(t-s)A}\left(\norm{\xi^{\mu_1}_s-\xi^{\mu_2}_s}+\sqrt{(\mu_1+\mu_2)\frac{B}{A}}\right).
\end{align}
\end{lemma}

\begin{proof}
To prove the first claim, assume that $\xi^\mu_t\in\Omega$ and $\xi^\mu_t(x)=2|r(x)|/\alpha(x,x)$, for some $x\in\Hn$. Then
\begin{align}
\tfrac{d}{dt}\xi^\mu_t(x)&\leq[r(x)-\alpha(x,x)\xi^\mu_t(x)]\xi^\mu_t(x)+\mu\sum_{y\sim x}b(y)m(y,x)\xi^\mu_t(y)\notag\\
&\leq\frac{-2|r(x)|^2}{\alpha(x,x)}+\mu2n\max_{y\in\Hn}\frac{b(y)|r(y)|}{\alpha(y,y)}<0,
\end{align}
for 
\begin{align}
\mu<\mu_0:=\min_{y\in\Hn}\frac{2|r(y)|^2}{\alpha(y,y)}\left(2n\max_{y\in\Hn}\frac{b(y)|r(y)|}{\alpha(y,y)}\right)^{-1}.
\end{align}
Hence, $\xi^\mu_t$ cannot leave $\Omega$.

For the second claim, we approximate
\begin{align}
\frac{d}{dt}\frac{\norm{\xi^{\mu_1}_t-\xi^{\mu_2}_t}^2}{2}
=&\ \sum_{x\in\Hn}(\xi^{\mu_1}_t(x)-\xi^{\mu_2}_t(x))r(x)(\xi^{\mu_1}_t(x)-\xi^{\mu_2}_t(x))\notag\\
&\ -\sum_{x\in\Hn}(\xi^{\mu_1}_t(x)-\xi^{\mu_2}_t(x))\sum_{y\in\Hn}\alpha(x,y)(\xi^{\mu_1}_t(x)\xi^{\mu_1}_t(y)-\xi^{\mu_2}_t(x)\xi^{\mu_2}_t(y))\notag\\
&\ +\sum_{x\in\Hn} (\xi^{\mu_1}_t(x)-\xi^{\mu_2}_t(x))\mu_1\left(\sum_{y\sim x}b(y)m(y,x)\xi^{\mu_1}_t(y)-b(x)\xi^{\mu_1}_t(x)\right)\notag\\
&\ -\sum_{x\in\Hn} (\xi^{\mu_1}_t(x)-\xi^{\mu_2}_t(x))\mu_2\left(\sum_{y\sim x}b(y)m(y,x)\xi^{\mu_2}_t(y)-b(x)\xi^{\mu_2}_t(x)\right)\notag\\
\leq&\ \max_{x\in\Hn}|r(x)|\sum_{x\in\Hn}(\xi^{\mu_1}_t(x)-\xi^{\mu_2}_t(x))^2\notag\\
&\ -\sum_{x\in\Hn}\sum_{y\in\Hn}\alpha(x,y)(\xi^{\mu_1}_t(x)-\xi^{\mu_2}_t(x))^2\xi^{\mu_1}_t(y)\notag\\
&\ +\sum_{x\in\Hn}\sum_{y\in\Hn}\alpha(x,y)|\xi^{\mu_1}_t(x)-\xi^{\mu_2}_t(x)|\cdot|\xi^{\mu_1}_t(y)-\xi^{\mu_2}_t(y)|\cdot|\xi^{\mu_2}_t(x)|\notag\\
&\ +\mu_1\max_{x\in\Hn}b(x)\sum_{x\in\Hn} \max_{x\in\Hn}(|\xi^{\mu_1}_t(x)|+|\xi^{\mu_2}_t(x)|)\max_{x\in\Hn}|\xi^{\mu_1}_t(x)|\left(\sum_{y\sim x}m(y,x)+1\right)\notag\\
&\ +\mu_2\max_{x\in\Hn}b(x)\sum_{x\in\Hn} \max_{x\in\Hn}(|\xi^{\mu_1}_t(x)|+|\xi^{\mu_2}_t(x)|)\max_{x\in\Hn}|\xi^{\mu_2}_t(x)|\left(\sum_{y\sim x}m(y,x)+1\right),
\end{align}
which implies
\begin{align}
\frac{d}{dt}\frac{\norm{\xi^{\mu_1}_t-\xi^{\mu_2}_t}^2}{2}
\leq&\ \norm{\xi^{\mu_1}_t-\xi^{\mu_2}_t}^2\Big[\max_{x\in\Hn}|r(x)|+2^{2n}\max_{x,y\in\Hn}\alpha(x,y)\norm{\xi^{\mu_2}_t}\Big]\notag\\
&\qquad\qquad\qquad+(\mu_1+\mu_2)(2^n\cdot2)\max_{x\in\Hn}b(x)(\norm{\xi^{\mu_1}_t}+\norm{\xi^{\mu_2}_t})^2\notag\\
=: &\ \norm{\xi^{\mu_1}_t-\xi^{\mu_2}_t}^2A+(\mu_1+\mu_2)B,
\end{align}
where $A$ and $B$ depend on $b,r,\alpha$, and can be chosen uniformly in $t\geq0$, $0\leq\mu_i<\mu_0$, and initial values $\xi^{\mu_i}_0\in\Omega$ since $\norm{\xi^{\mu_i}_t}\leq\max_{\xi\in\Omega}\norm{\xi}<\infty$. Applying Gronwall's inequality and taking the square root implies the claim.
\end{proof}


Theorem \ref{Thmexp} and Corollary \ref{Corexp} provide us with approximations for $\xi^\mu_t$ during the exponential growth phase and Theorem \ref{LVThm} guarantees convergence to a new equilibrium during the invasion phase. To show that this second phase vanishes on the time scale $\ln1/\mu$, we need to bound its duration uniformly in the approximate state of the system at its beginning.

We introduce the following notation for the time until the initial conditions for the next growth phase are reached.

\begin{definition}
\begin{align}
\tilde{\tau}^\mu_\eta(\xi, \x):=\inf\big\{t\geq0:&\ \forall\ x\in\x: |\xi^\mu_t(x)-\bar{\xi}_\x(x)|\leq \eta\frac{\bar{c}}{\sqrt{|\x|}},\notag\\
&\ \forall\ y\in\Hn\backslash\x: \xi^\mu_t(y)\leq \frac{\eta}{3}; \xi^\mu_0=\xi\big\},
\end{align}
\end{definition}

In the proof of Theorem \ref{MainThm}, we approximate the true system, solving (\ref{DE}), by the mutation-free Lotka-Volterra system during the invasion. The second lemma proves continuity in the initial condition for a slight variation of $\tilde{\tau}^\mu_\eta(\xi,\x)$, corresponding to the case of $\mu=0$.

\begin{lemma}\label{conttau}
Let $\y\subset\Hn$ such that $r(y)>0$, for all $y\in\y$, and (B$_\y$) is satisfied. Let $\x\subset\y$ such that the equilibrium state of the Lotka-Volterra system involving types $\y$ is supported on $\x$ and assume $f_{y,\x}<0$, for every $y\in\y\backslash\x$. Define
\begin{align}
\bar{\tau}^0_\eta(\xi,\x,\y):=\inf\{t\geq0:&\ \norm{\left.\xi^0_t\right|_\x-\bar{\xi}_\x}_\x\leq\frac{\eta\bar{c}c_\x}{2\sqrt{|\x|}},\notag\\
&\ \forall\ y\in\y\backslash\x: \xi^0_t(y)\leq\frac{\eta}{6}\land\hat{\eta};\xi^0_0=\xi\},
\end{align}
where $\norm{\cdot}_\x$ is the norm defined in (\ref{defnorm}), corresponding to $\bar{\xi}_\x$, and $\hat{\eta}:=\eta\bar{c}c_\x/(2\sqrt{|\x|}\underline{c})$. Then, for $\eta$ small enough, $\bar{\tau}^0_\eta(\xi,\x,\y)$ is continuous in $\xi\in(\R_{>0})^\y\times\{0\}^{\Hn\backslash\y}$.
\end{lemma}

\begin{remark}
Theorem \ref{LVThm} ensures that the Lotka-Volterra system involving the types $\y$ converges to a unique equilibrium and hence $\x$ in Lemma \ref{conttau} is uniquely determined.
\end{remark}

\begin{proof}
Since we are considering the case of $\mu=0$, we obtain $\xi^0_t\in(\R_{>0})^\y\times\{0\}^{\Hn\backslash\y}$, for all $t\geq 0$ and $\xi^0_0\in(\R_{>0})^\y\times\{0\}^{\Hn\backslash\y}$. As in Step 1 of the proof of Theorem \ref{Thmexp}, it follows that, as long as $\xi^0_t(y)\leq \hat{\eta}$, for $y\in\y\backslash\x$, and
\begin{align}
\hat{\eta} \underline{c}\leq \norm{\left.\xi^0_t\right|_\x-\bar{\xi}_{\x}}_\x\leq \eps_{\x},\label{disc}
\end{align}
we obtain
\begin{align}
\frac{d}{dt}\frac{\norm{\left.\xi^0_t\right|_\x-\bar{\xi}_{\x}}_\x^2}{2}\leq -\norm{\left.\xi^0_t\right|_\x-\bar{\xi}_{\x}}_\x^2\frac{\kappa}{4C_\x^2}=:-\tilde{\kappa}\norm{\left.\xi^0_t\right|_\x-\bar{\xi}_{\x}}_\x^2.
\end{align}
Hence
\begin{align}
\norm{\left.\xi^0_t\right|_\x-\bar{\xi}_{\x}}_\x\leq \ee^{-\tilde{\kappa}(t-t_0)}\norm{\left.\xi^0_{t_0}\right|_\x-\bar{\xi}_{\x}}_\x.\label{decresident}
\end{align}

Moreover, (\ref{disc}) implies, for every $x\in\x$,
\begin{align}
|\xi^0_t(x)-\bar{\xi}_{\x}(x)|\leq\frac{\eps_\x}{c_\x}.
\end{align}
Since $f_{y,\x}<0$ for every $y\in\y\backslash\x$, we can choose $\eps_\x$ small enough such that
\begin{align}
\tfrac{d}{dt}\xi^0_t(y)&=[r(y)-\sum_{z\in\Hn}\alpha(y,z)\xi^0_t(z)]\xi^0_t(y)\notag\\
&\leq\left[f_{y,\x}+\sum_{x\in\x}\alpha(y,x)\frac{\eps_\x}{c_\x}\right]\xi^0_t(y)\leq -C\xi^0_t(y),
\end{align}
for some $C>0$. Hence,
\begin{align}
\xi^0_t(y)\leq \ee^{-C(t-t_0)}\xi^0_{t_0}(y).\label{decrest}
\end{align}

We have now found an attractive domain around the limiting equilibrium of the Lotka-Volterra system.

Next, we can derive the continuity of $\bar{\tau}^0_\eta(\xi,\x,\y)$. Let $\gamma>0$ arbitrarily small such that $\ee^{\tilde{\kappa}\gamma},\ee^{C\gamma}\leq 2$. Let $\xi^{0,1}$ and $\xi^{0,2}$ be two versions of the process with different initial values $\xi^{0,1}_0$ and $\xi^{0,2}_0$. By Lemma \ref{cont},
\begin{align}
&\norm{\left.\xi^{0,1}_t\right|_\x-\left.\xi^{0,2}_t\right|_\x}_\x\leq C_\x\norm{\left.\xi^{0,1}_t\right|_\x-\left.\xi^{0,2}_t\right|_\x}\leq \ee^{(t-t_0)A}C_\x\norm{\left.\xi^{0,1}_{t_0}\right|_\x-\left.\xi^{0,2}_{t_0}\right|_\x},\\
&\ |\xi^{0,1}_t(y)-\xi^{0,2}_t(y)|\leq\norm{\xi^{0,1}_t-\xi^{0,2}_t}\leq \ee^{(t-t_0)A}\norm{\xi^{0,1}_{t_0}-\xi^{0,2}_{t_0}}.
\end{align}

Now, if we pick initial conditions that are very similar, namely that satisfy
\begin{align}
\norm{\xi^{0,1}_0-\xi^{0,2}_0}\leq \ee^{-(\bar{\tau}^0_{\bar{\eta}}(\xi^{0,1}_{t_0},\x,\y)+\gamma)A}\left[(\ee^{\tilde{\kappa}\gamma}-1)\frac{\eta\bar{c}c_\x}{2\sqrt{|\x|}C_\x}\land(\ee^{C\gamma}-1)\left(\frac{\eta}{6}\land\hat{\eta}\right)\right],
\end{align}
we can apply (4.15) and (4.16) and use the definition of $\bar{\tau}^0_{\eta}(\xi^{0,1}_0,\x,\y)$ to derive
\begin{align}
\norm{\left.\xi^{0,2}_{\bar{\tau}^0_{\eta}(\xi^{0,1}_0,\x,\y)}\right|_\x-\bar{\xi}_{\x}}_\x
&\leq \norm{\left.\xi^{0,2}_{\bar{\tau}^0_{\eta}(\xi^{0,1}_0,\x,\y)}\right|_\x-\left.\xi^{0,1}_{\bar{\tau}^0_{\eta}(\xi^{0,1}_0,\x,\y)}\right|_\x}_\x+\norm{\left.\xi^{0,1}_{\bar{\tau}^0_{\eta}(\xi^{0,1}_0,\x,\y)}\right|_\x-\bar{\xi}_{\x}}_\x\notag\\
&\leq \ee^{\bar{\tau}^0_{\eta}(\xi^{0,1}_0,\x,\y)A}C_\x\norm{\left.\xi^{0,2}_0\right|_\x-\left.\xi^{0,1}_0\right|_\x}+\frac{\eta\bar{c}c_\x}{2\sqrt{|\x|}}
\leq \ee^{\tilde{\kappa}\gamma}\frac{\eta\bar{c}c_\x}{2\sqrt{|\x|}},
\end{align}
and for $y\in\y\backslash\x$,
\begin{align}
\xi^{0,2}_{\bar{\tau}^0_{\eta}(\xi^{0,1}_0,\x,\y)}(y)
&\leq |\xi^{0,2}_{\bar{\tau}^0_{\eta}(\xi^{0,1}_0,\x,\y)}(y)-\xi^{0,1}_{\bar{\tau}^0_{\eta}(\xi^{0,1}_0,\x,\y)}(y)|+\xi^{0,1}_{\bar{\tau}^0_{\eta}(\xi^{0,1}_0,\x,\y)}(y)\notag\\
&\leq \ee^{\bar{\tau}^0_{\eta}(\xi^{0,1}_0,\x,\y)A}\norm{\xi^{0,2}_0-\xi^{0,1}_0}+\left(\frac{\eta}{6}\land\hat{\eta}\right)
\leq \ee^{C\gamma}\left(\frac{\eta}{6}\land\hat{\eta}\right).
\end{align}

For all $\eta>0$ such that
\begin{align}
\hat{\eta}\underline{c}=\frac{\eta\bar{c}c_\x}{2\sqrt{|\x|}}\leq \frac{\eps_\x}{2},
\end{align}
we obtain
\begin{align}
\norm{\left.\xi^{0,2}_{\bar{\tau}^0_{\eta}(\xi^{0,1}_0,\x,\y)}\right|_\x-\bar{\xi}_{\x}}_\x\leq\eps_\x
\end{align}
and hence (\ref{decresident}) and (\ref{decrest}) can be applied to $\xi^{0,2}$ with $t=\bar{\tau}^0_{\eta}(\xi^{0,1}_0,\x,\y)+\gamma$ and $t_0=\bar{\tau}^0_{\eta}(\xi^{0,1}_0,\x,\y)$ to obtain $\bar{\tau}^0_{\eta}(\xi^{0,2}_0,\x,\y)\leq\bar{\tau}^0_{\eta}(\xi^{0,1}_0,\x,\y)+\gamma$.

Repeating the same calculation switching 1 and 2 and using this bound for $\bar{\tau}^0_{\eta}(\xi^{0,2}_0,\x,\y)$ to apply (4.17), it follows that
\begin{align}
&\norm{\left.\xi^{0,1}_{\bar{\tau}^0_{\eta}(\xi^{0,2}_0,\x,\y)}\right|_\x-\bar{\xi}_{\x}}_\x\leq \ee^{\tilde{\kappa}\gamma}\frac{\eta\bar{c}c_\x}{2\sqrt{|\x|}},\\
&\xi^{0,1}_{\bar{\tau}^0_{\eta}(\xi^{0,2}_0,\x,\y)}(y)\leq \ee^{C\gamma}\left(\frac{\eta}{6}\land\hat{\eta}\right),
\end{align}
and therefore $\bar{\tau}^0_{\eta}(\xi^{0,1}_0,\x,\y)\leq\bar{\tau}^0_{\eta}(\xi^{0,2}_0,\x,\y)+\gamma$. Hence, $|\bar{\tau}^0_{\eta}(\xi^{0,1}_0,\x,\y)-\bar{\tau}^0_{\eta}(\xi^{0,2}_0,\x,\y)|\leq\gamma$, which proves the continuity.
\end{proof}


To mark the transition between the exponential growth phase and the Lotka-Volterra invasion phase, we extend the definition of $\tilde{T}^\mu_\eta$ (Definition 7) to the $i^\text{th}$ invasion.
\begin{definition}
For $i\geq1$, the time when the first mutant type reaches $\eta>0$ after the $(i-1)^\text{st}$ invasion is defined as
\begin{align}
\tilde{T}^\mu_{\eta,i}:=\inf\{s\geq \tilde{T}^\mu_{\eta,i-1}:\exists\ y\in\Hn\backslash(\x^{i-2}\cup\x^{i-1}):\xi^\mu_s(y)>\eta\}.
\end{align}
We set $\tilde{T}^\mu_{\eta,0}:=0$ and $\x^{-1}:=\emptyset$.

To consider the evolutionary time scale $\ln1/\mu$, we define $T^\mu_{\eta,i}$ through $\tilde{T}^\mu_{\eta,i}=T^\mu_{\eta,i}\ln1/\mu$.
\end{definition}

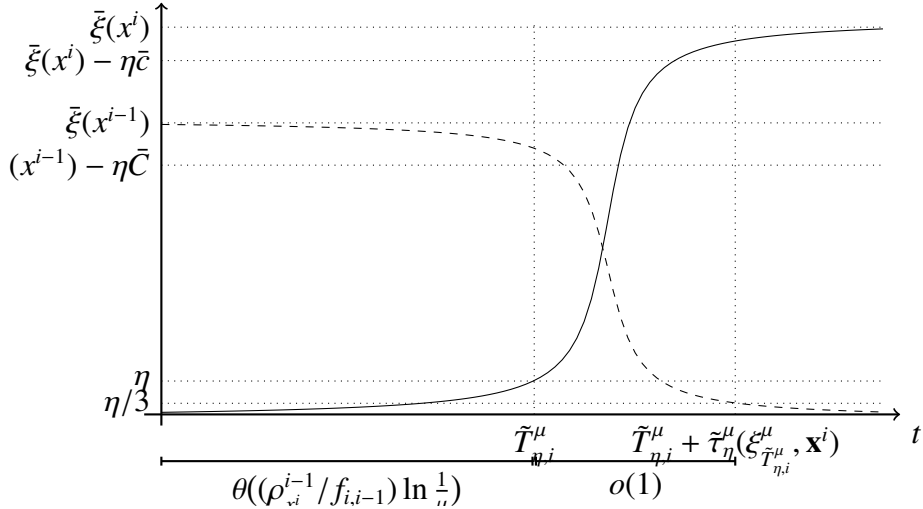
\begin{figure}[h]
\centering
\begin{tikzpicture}
\begin{axis}[xscale=1.9,yscale=1.25,hide axis]
\addplot [domain=-13:8, samples=100,]{atan(1.5*x)+3};
\addplot [domain=-13:8, samples=100,dashed]{-0.75*atan(1.5*x)-20};
\addplot [domain=-15:9, samples=100, color=white,]{95}; 
\addplot [domain=-15:9, samples=100, color=white,]{-95}; 
\draw[thick,->](axis cs:-13.5,-85) -- (axis cs:8.5,-85); 
	\addplot[mark=none] coordinates {(8.5,-85)} node[anchor=north west]{$t$};
\draw[thick,->](axis cs:-13,-90)--(axis cs:-13,100); 
\draw[dotted](axis cs:-13,-70) -- (axis cs:8,-70);
	\addplot[mark=none] coordinates {(-13,-70)} node[anchor=east]{$\eta$};
\draw[dotted](axis cs:-13,-80) -- (axis cs:8,-80);
	\addplot[mark=none] coordinates {(-13,-80)} node[anchor=east]{$\eta/3$};
\draw[dotted](axis cs:-13,46) -- (axis cs:8,46);
	\addplot[mark=none] coordinates {(-13,46)} node[anchor=east]{$\bar{\xi}(x^{i-1})$};
\draw[dotted](axis cs:-13,27) -- (axis cs:8,27);
	\addplot[mark=none] coordinates {(-13,27)} node[anchor=east]{$\bar{\xi}(x^{i-1})-\eta\bar{C}$};
\draw[dotted](axis cs:-13,89) -- (axis cs:8,89);
	\addplot[mark=none] coordinates {(-13,89)} node[anchor=east]{$\bar{\xi}(x^i)$};
\draw[dotted](axis cs:-13,74) -- (axis cs:8,74);
	\addplot[mark=none] coordinates {(-13,74)} node[anchor=east]{$\bar{\xi}(x^i)-\eta\bar{c}$};
\draw[dotted](axis cs:-2.15,-85) -- (axis cs:-2.15,90);
	\addplot[mark=none] coordinates {(-2.15,-85)} node[anchor=north]{$\tilde{T}^\mu_{\eta,i}$};
\draw[dotted](axis cs:3.7,-85) -- (axis cs:3.7,90);
	\addplot[mark=none] coordinates {(3.7,-85)} node[anchor=north]{$\tilde{T}^\mu_{\eta,i}+\tilde{\tau}^\mu_\eta(\xi^\mu_{\tilde{T}^\mu_{\eta,i}},\x^i)$};
\draw[thick](axis cs:-13,-106)--(axis cs:-2.2,-106);
	\draw[thick](axis cs:-13,-108)--(axis cs:-13,-104);
	\draw[thick](axis cs:-2.2,-108)--(axis cs:-2.2,-104);
	\addplot[mark=none] coordinates {(-7.6,-106)} node[anchor=north]{$\theta\big((\rho^{i-1}_{x^i}/f_{i,i-1})\ln\frac{1}{\mu}\big)$};
\draw[thick](axis cs:-2.1,-106)--(axis cs:3.7,-106);
	\draw[thick](axis cs:-2.1,-108)--(axis cs:-2.1,-104);
	\draw[thick](axis cs:3.7,-108)--(axis cs:3.7,-104);
	\addplot[mark=none] coordinates {(0.8,-106)} node[anchor=north]{$o(1)$};
\end{axis}
\end{tikzpicture}
\caption{The two phases of $y^i_*=x^i$ invading $x^{i-1}$, in the case where there is no coexistence. The dashed line corresponds to $\xi^\mu_t(x^{i-1})$, the solid line depicts $\xi^\mu_t(x^i)$.}
\end{figure}

We can now turn to the proof of Theorem \ref{MainThm} and inductively derive the convergence of $\xi^\mu_{t\ln 1/\mu}$ to a jump process as $\mu\to0$. The two phases of an invasion (exponential growth and Lotka-Volterra) are depicted in Figure 2.


\begin{proof}[Theorem \ref{MainThm}]
The proof is split into several parts. The main goal is to inductively approximate $T^\mu_{\eta,i}$ and $\xi^\mu_{t\ln1/\mu}$, similar to Corollary \ref{Corexp}.
We claim that, for each $1\leq i\leq I$ such that $T_i<\infty$,
\begin{align}\label{ClaimTime}
\min_{\substack{y\in\Hn\\ \rho^{i-1}_y>0}}\min_{\substack{z\in\Hn\\f_{z,\x^{i-1}}>0}}\frac{\rho^{i-1}_z+|z-y|-\eta\hat{C}_{i-1}}{f_{z,\x^{i-1}}+\eta\hat{C}}\leq\liminf_{\mu\to0}T^\mu_{\eta,i}-T_{i-1}\notag\\
\leq\limsup_{\mu\to0}T^\mu_{\eta,i}-T_{i-1}\leq\min_{\substack{y\in\Hn\\ \rho^{i-1}_y>0}}\min_{\substack{z\in\Hn\\f_{z,\x^{i-1}}>0}}\frac{\rho^{i-1}_z+|z-y|+\eta\check{C}_{i-1}}{f_{z,\x^{i-1}}-\eta\check{C}}.
\end{align}
Moreover, for each $0\leq i<I$ such that $T_i<\infty$, $T_i<t<T_{i+1}$, there are positive constants $\check{c}_i$, $\check{C}_i$, $\hat{c}_i$, $\hat{C}_i$, and $m$, such that, for every $y\in\Hn$,
\begin{align}\label{ClaimSize}
\check{c}_i&\mu^{\min_{z\in\Hn}[\rho^i_z+|z-y|-(t-T_i)(f_{z,\x^i}-\eta\check{C})]+\eta\check{C}_i}\leq\xi^\mu_{t\ln1/\mu}(y)\notag\\
&\leq\hat{c}_i\mu^{\min_{z\in\Hn}[\rho^i_z+|z-y|-(t-T_i)(f_{z,\x^i}+\eta\hat{C})]-\eta\hat{C}_i}\left(1+t\ln\frac{1}{\mu}\right)^{(i+1)m},
\end{align}
while, for each $x\in\x^i$, $\xi^\mu_{t\ln1/\mu}(x)\in[\bar{\xi}_{\x^i}(x)-\eta\bar{C},\bar{\xi}_{\x^i}(x)+\eta\bar{C}]$.

In the first step, we approximate $|T^\mu_{\eta,i}-T_i|\leq \eta C$, assuming that the claim holds true. Second, we derive a uniform bound on the duration of the $i^\text{th}$ invasion phase, using Lemma \ref{conttau}. In Step 3, we prove the bounds that are claimed above. Finally, we use these bounds to derive the convergence as $\mu\to0$.

\noindent\textit{Step 1:} $|T^\mu_{\eta,i}-T_i|\leq \eta C$.\\
In the case where there exists a $y\in\Hn$ such that $f_{y,x^{i-1}}>0$, we want to relate $T_i$, as defined in (\ref{defTi}), to $T^\mu_{\eta,i}$.

First, we prove a different identity for $T_i$ that is similar to (\ref{ClaimTime}), namely the second equality of
\begin{align}
T_i-T_{i-1}=\min_{\substack{y\in\Hn:\\f_{y,\x^{i-1}}>0}}\frac{\rho_y^{i-1}}{f_{y,\x^{i-1}}}=\min_{\substack{y\in\Hn\\\rho^{i-1}_y>0}}\min_{\substack{z\in\Hn\\f_{z,\x^{i-1}}>0}}\frac{\rho^{i-1}_z+|z-y|}{f_{z,\x^{i-1}}}.\label{jumptime}
\end{align}

On one hand, $f_{y,\x^{i-1}}>0$ implies $\rho^{i-1}_y>0$. The only cases in which $\rho^{i-1}_y=0$ are if $y\in\x^{i-1}$, then $f_{y,\x^{i-1}}=0$, or if $y\in\x^{i-2}\backslash\x^{i-1}$, which implies $f_{y,\x^{i-1}}<0$ (else we would have terminated the procedure after the $(i-1)^\text{st}$ invasion due to case (b) in Theorem \ref{MainThm}). Hence
\begin{align}
\min_{\substack{y\in\Hn\\\rho^{i-1}_y>0}}\min_{\substack{z\in\Hn\\f_{z,\x^{i-1}}>0}}\frac{\rho^{i-1}_z+|z-y|}{f_{z,\x^{i-1}}}
\leq\min_{\substack{y\in\Hn\\f_{y,\x^{i-1}}>0}}\min_{\substack{z\in\Hn\\f_{z,\x^{i-1}}>0}}\frac{\rho^{i-1}_z+|z-y|}{f_{z,\x^{i-1}}}
\leq\min_{\substack{y\in\Hn\\f_{y,\x^{i-1}}>0}}\frac{\rho^{i-1}_y}{f_{y,\x^{i-1}}},
\end{align}
where we inserted $z=y$ in the second step.

On the other hand, if we assume that $\bar{y}$ and $\bar{z}$ realise the minima, which implies that \linebreak $f_{\bar{z},\x^{i-1}}>0$, we obtain
\begin{align}
\min_{\substack{y\in\Hn\\\rho^{i-1}_y>0}}\min_{\substack{z\in\Hn\\f_{z,\x^{i-1}}>0}}\frac{\rho^{i-1}_z+|z-y|}{f_{z,\x^{i-1}}}
=\frac{\rho^{i-1}_{\bar{z}}+|\bar{z}-\bar{y}|}{f_{\bar{z},\x^{i-1}}}
\geq \frac{\rho^{i-1}_{\bar{z}}}{f_{\bar{z},\x^{i-1}}}
\geq\min_{\substack{y\in\Hn\\f_{y,\x^{i-1}}>0}}\frac{\rho^{i-1}_y}{f_{y,\x^{i-1}}}.
\end{align}

Now, under the assumption that (\ref{ClaimTime}) holds true, we approximate
\begin{align}
\liminf_{\mu\to0}T^\mu_{\eta,i}-T_{i-1}\geq&\ \left(\min_{\substack{y\in\Hn\\ \rho^{i-1}_y>0}}\min_{\substack{z\in\Hn\\f_{z,\x^{i-1}}>0}}\frac{\rho^{i-1}_z+|z-y|}{f_{z,\x^{i-1}}}\right)\left(\min_{\substack{z\in\Hn\\f_{z,\x^{i-1}}>0}}\frac{f_{z,\x^{i-1}}}{f_{z,\x^{i-1}}+\eta\hat{C}}\right)\notag\\
&\ -\eta\hat{C}_{i-1}\max_{\substack{z\in\Hn\\f_{z,\x^{i-1}}>0}}\frac{1}{f_{z,\x^{i-1}}+\eta\hat{C}}\notag\\
=&\ (T_i-T_{i-1})\left(1-\max_{\substack{z\in\Hn\\f_{z,\x^{i-1}}>0}}\frac{\eta\hat{C}}{f_{z,\x^{i-1}}+\eta\hat{C}}\right)-\eta\hat{C}_{i-1}\max_{\substack{z\in\Hn\\f_{z,\x^{i-1}}>0}}\frac{1}{f_{z,\x^{i-1}}+\eta\hat{C}}\notag\\
=&\ (T_i-T_{i-1})-\eta((T_i-T_{i-1})\hat{C}+\hat{C}_{i-1})\max_{\substack{z\in\Hn\\f_{z,\x^{i-1}}>0}}\frac{1}{f_{z,\x^{i-1}}+\eta\hat{C}}
\end{align}
and, analogously,
\begin{align}
\limsup_{\mu\to0}T^\mu_{\eta,i}-T_{i-1}\leq&\ (T_i-T_{i-1})+\eta((T_i-T_{i-1})\check{C}+\check{C}_{i-1})\max_{\substack{z\in\Hn\\f_{z,\x^{i-1}}>0}}\frac{1}{f_{z,\x^{i-1}}-\eta\check{C}}.
\end{align}
As a result there is a constant $C>0$ such that, for $\eta$ and $\mu$ small enough,
\begin{align}
|T^\mu_{\eta,i}-T_i|\leq \eta C.
\end{align}


\noindent\textit{Step 2:} Uniform time bound on the Lotka-Volterra phase.\\
We show that, for $\eta$ small enough,
\begin{align}
\tilde{\tau}^\mu_\eta(\xi^\mu_{\tilde{T}^\mu_{\eta,i}},\x^i)=\inf\big\{t\geq0:&\ \forall\ x\in\x^i: |\xi^\mu_{\tilde{T}^\mu_{\eta,i}+t}(x)-\bar{\xi}_{\x^i}(x)|\leq \eta\frac{\bar{c}}{\sqrt{|\x^i|}},\notag\\
&\ \forall\ y\in\Hn\backslash\x^i: \xi^\mu_{\tilde{T}^\mu_{\eta,i}+t}(y)\leq \frac{\eta}{3}\big\}
\end{align}
is bounded by some constant $\bar{T}_\eta$.

Since $\text{LVE}_+(\x^{i-1})=\{\bar{\xi}_{\x^{i-1}}\}$ and $f_{y^i_*,\x^{i-1}}>0$, we obtain $r(y)>0$, for every $y\in(\x^{i-1}\cup y^i_*)$. (B$_{\x^{i-1}\cup y^i_*}$) holds by assumption and hence Lemma \ref{conttau} can be applied to $\y=\x^{i-1}\cup y^i_*$ and $\x=\x^i$.

Let 
\begin{align}
\Omega^i_\eta:=\{\xi: \xi(y^i_*)=\eta,\xi(x)\in[\bar{\xi}_{\x^{i-1}}(x)-\eta\bar{C},\bar{\xi}_{\x^{i-1}}(x)+\eta\bar{C}]\ \forall\ x\in\x^{i-1}, \xi(y)=0\text{ else}\},
\end{align}
then, by continuity of $\bar{\tau}^0_\eta(\xi,\x^i,\x^{i-1}\cup y^i_*)$ in $\xi$ (Lemma \ref{conttau}) and the compactness of $\Omega^i_\eta$,
\begin{align}
\sup_{\xi\in\Omega^i_\eta}\bar{\tau}^0_{\eta}(\xi,\x^i,\x^{i-1}\cup y^i_*)=:\bar{T}_\eta<\infty.
\end{align}

Using Lemma \ref{cont}, for
\begin{align}
\xi:=\begin{cases}\xi^\mu_{\tilde{T}^\mu_{\eta,i}}(x)&x\in\x^{i-1}\cup y^i_*\\0&\text{else}\end{cases}\ \in\Omega^i_\eta,\qquad\bar{\tau}:=\bar{\tau}^0_{\eta}(\xi,\x^i,\x^{i-1}\cup y^i_*),
\end{align}
we obtain, for $x\in\x^i$, $y\in\x^{i-1}\cup y^i_*\backslash\x^i$, $\xi^0_0=\xi$, and $\mu$ small enough, that
\begin{align}
|\xi^\mu_{\tilde{T}^\mu_{\eta,i}+\bar{\tau}}(x)-\bar{\xi}_{\x^i}(x)|&\leq\norm{\xi^\mu_{\tilde{T}^\mu_{\eta,i}+\bar{\tau}}-\xi^0_{\bar{\tau}}}+c_{\x^i}^{-1}\norm{\left.\xi^0_{\bar{\tau}}\right|_{\x^i}-\bar{\xi}_{\x^i}}_{\x^i}\notag\\
&\leq \ee^{\bar{\tau} A}\left(\norm{\xi^\mu_{\tilde{T}^\mu_{\eta,i}}-\xi}+\sqrt{\mu\frac{B}{A}}\right)+\frac{\eta\bar{c}}{2\sqrt{|\x^i|}}\leq\frac{\eta\bar{c}}{\sqrt{|\x^i|}},\\
\xi^\mu_{\tilde{T}^\mu_{\eta,i}+\bar{\tau}}(y)&\leq \norm{\xi^\mu_{\tilde{T}^\mu_{\eta,i}+\bar{\tau}}-\xi^0_{\bar{\tau}}}+\xi^0_{\bar{\tau}}(y)\notag\\
&\leq \ee^{{\bar{\tau}} A}\left(\norm{\xi^\mu_{\tilde{T}^\mu_{\eta,i}}-\xi}+\sqrt{\mu\frac{B}{A}}\right)+\frac{\eta}{6}\leq\frac{\eta}{3}.
\end{align}
Here we used that, for $\eta$ small enough, $\norm{\xi^\mu_{\tilde{T}^\mu_{\eta,i}}-\xi}\leq 2^n\max_{y\in\Hn\backslash(\x^{i-1}\cup y^i_*)}\xi^\mu_{\tilde{T}^\mu_{\eta,i}}(y)$ tends to zero as $\mu\to0$. A more precise approximation for this is given in Step 3 and 4.

Overall, $\tilde{\tau}^\mu_{\eta}(\xi^\mu_{\tilde{T}^\mu_{\eta,i}},\x^i)\leq \bar{\tau}\leq \bar{T}_\eta$.


\noindent\textit{Step 3:} Approximation of $\xi^\mu_{t\ln1/\mu}$ and $T^\mu_{\eta,i}$.\\
We now turn to the proof of (\ref{ClaimTime}) and (\ref{ClaimSize}).

(\ref{ClaimSize}) in the case of $i=0$ is given by Theorem \ref{Thmexp} and Corollary \ref{Corexp}, setting $\check{c}_0:=\check{c}$, $\check{C}_i:=0$, $\hat{c}_0:=2^n\hat{c}$, and $\hat{C}_i:=0$ and using that by Step 1, for every $t<T_1$, there are $\eta$ and $\mu$ small enough such that $t<T^\mu_{\eta,1}$. Corollary \ref{Corexp} also gives(\ref{ClaimTime}) for $i=1$.

Assuming that the claims holds for $0\leq i-1<I$, $T_i<\infty$ implies that there is some $y'\in\Hn$ for which $f_{y',\x^{i-1}}>0$, and hence, for every $y\in\Hn$,
\begin{align}
\check{c}_{i-1}&\mu^{\min_{z\in\Hn}[\rho^{i-1}_z+|z-y|-(T^\mu_{\eta,i}-T_{i-1})(f_{z,\x^{i-1}}-\eta\check{C})]+\eta\check{C}_{i-1}}\leq\xi^\mu_{\tilde{T}^\mu_{\eta,i}}(y)\notag\\
&\leq\hat{c}_{i-1}\mu^{\min_{z\in\Hn}[\rho^{i-1}_z+|z-y|-(T^\mu_{\eta,i}-T_{i-1})(f_{z,\x^{i-1}}+\eta\hat{C})]-\eta\hat{C}_{i-1}}\left(1+\tilde{T}^\mu_{\eta,i}\right)^{im}.
\end{align}
Moreover, $\xi^\mu_{\tilde{T}^\mu_{\eta,i}}(y^i_*)=\eta$ and, for every $x\in\x^{i-1}$, $\xi^\mu_{\tilde{T}^\mu_{\eta,i}}(x)\in[\bar{\xi}_{\x^{i-1}}(x)-\eta\bar{C},\bar{\xi}_{\x^{i-1}}(x)+\eta\bar{C}]$. Similar to Corollary \ref{Corexp}, we obtain (\ref{ClaimTime}).

Next, we estimate the evolution of the different types during the Lotka-Volterra phase. Lemma \ref{cont} gives $\xi^\mu_t(z)\leq2|r(z)|/\alpha(z,z)$, for all $z\in\Hn$ and $t\geq 0$, and therefore
\begin{align}
\tfrac{d}{dt}\xi^\mu_t(y)\geq\left[r(y)-\sum_{z\in\Hn}\alpha(y,z)\frac{2|r(z)|}{\alpha(z,z)}-\mu b(y)\right]\xi^\mu_t(y)\geq -K\xi^\mu_t(y),
\end{align}
for some $K>0$.

By Step 2, we know that $\tilde{\tau}(\xi^\mu_{\tilde{T}^\mu_{\eta,i}},\x^i)\leq \bar{T}_\eta$ and hence (4.39) yields
\begin{align}
\xi^\mu_{\tilde{T}^\mu_{\eta,i}+\tilde{\tau}(\xi^\mu_{\tilde{T}^\mu_{\eta,i}},\x^i)}(y)\geq \ee^{-K\bar{T}_\eta}\check{c}_{i-1}\mu^{\min_{z\in\Hn}[\rho^{i-1}_z+|z-y|-(T^\mu_{\eta,i}-T_{i-1})(f_{z,\x^{i-1}}-\eta\check{C})]+\eta\check{C}_{i-1}}
\end{align}

Using Step 1, we can approximate
\begin{align}
\min_{z\in\Hn}& [\rho^{i-1}_z+|z-y|-(T^\mu_{\eta,i}-T_{i-1})(f_{z,\x^{i-1}}-\eta\check{C})]+\eta\check{C}_{i-1}\notag\\
=&\ \min_{z\in\Hn}[\rho^{i-1}_z+|z-y|-(T^\mu_{\eta,i}-T_{i-1})f_{z,\x^{i-1}}]+\eta(\check{C}_{i-1}+(T^\mu_{\eta,i}-T_{i-1})\check{C})\notag\\
\leq&\ \rho^i_y+\eta(\check{C}_{i-1}+(T^\mu_{\eta,i}-T_{i-1})\check{C}+C\max_{z\in\Hn}f_{z,\x^{i-1}}).
\end{align}

We now plug this back in as the exponent and set $\check{c}'_i:=\ee^{-K\bar{T}_\eta}\check{c}_{i-1}$ as well as \linebreak $\check{C}'_i\geq\check{C}_{i-1}+(T^\mu_{\eta,i}-T_{i-1})\check{C}+C\max_{z\in\Hn}f_{z,\x^{i-1}}$ to derive
\begin{align}
\xi^\mu_{\tilde{T}^\mu_{\eta,i}+\tilde{\tau}(\xi^\mu_{\tilde{T}^\mu_{\eta,i}},\x^i)}(y)\geq\check{c}'_i\mu^{\rho^i_y+\eta\hat{C}'_i} .
\end{align}
Note that $\check{C}'_i$ can be chosen uniformly in $\eta$ since $T^\mu_{\eta,i}\leq T_i+\eta C$ by Step 1, while $\check{c}'_i$ may depend on $\eta$.

On the other hand,
\begin{align}
\tfrac{d}{dt}\xi^\mu_t(y)\leq r(y)\xi^\mu_t(y)+\mu\tilde{C}\sum_{z\sim y}\xi^\mu_t(z).
\end{align}

Following the same argument as for the upper bound in (\ref{preexp}) (compare Step 2 of the proof of Theorem \ref{Thmexp}, with $t=\tilde{\tau}(\xi^\mu_{\tilde{T}^\mu_{\eta,i}},\x^i)$ and $\xi^\mu_{\tilde{T}^\mu_{\eta,i}}$ instead of $\xi^\mu_0$), we obtain
\begin{align}
\xi^\mu_{\tilde{T}^\mu_{\eta,i}+\tilde{\tau}(\xi^\mu_{\tilde{T}^\mu_{\eta,i}},\x^i)}(y)&\leq \hat{c} \ee^{\tilde{\tau}(\xi^\mu_{\tilde{T}^\mu_{\eta,i}},\x^i)\max_{z\in\Hn}r(z)}(1+\tilde{\tau}(\xi^\mu_{\tilde{T}^\mu_{\eta,i}},\x^i))^m\sum_{z\in\Hn}\xi^\mu_{\tilde{T}^\mu_{\eta,i}}(z)\mu^{|z-y|}.
\end{align}
By Step 1,
\begin{align}
\min_{z'\in\Hn}& [\rho^{i-1}_{z'}+|z'-z|-(T^\mu_{\eta,i}-T_{i-1})(f_{z',\x^{i-1}}+\eta\hat{C})]-\eta\hat{C}_{i-1}+|z-y|\notag\\
\geq&\ \min_{z'\in\Hn}[\rho^{i-1}_{z'}+|z'-y|-(T^\mu_{\eta,i}-T_{i-1})f_{z',\x^{i-1}}]-\eta(\hat{C}_{i-1}+(T^\mu_{\eta,i}-T_{i-1})\hat{C})\notag\\
\geq&\ \rho^i_y-\eta(\hat{C}_{i-1}+(T^\mu_{\eta,i}-T_{i-1})\hat{C}+C\max_{z\in\Hn}f_{z,\x^{i-1}}).
\end{align}
Using this and Step 2, we derive
\begin{align}
\xi^\mu_{\tilde{T}^\mu_{\eta,i}+\tilde{\tau}(\xi^\mu_{\tilde{T}^\mu_{\eta,i}},\x^i)}(y)&\leq\hat{c} \ee^{\bar{T}_\eta\max_{z\in\Hn}r(z)}(1+\bar{T}_\eta)^m\notag\\
&\qquad\cdot\sum_{z\in\Hn}\hat{c}_{i-1}\mu^{\rho^i_y-\eta(\hat{C}_{i-1}+(T^\mu_{\eta,i}-T_{i-1})\hat{C}+C\max_{z\in\Hn}f_{z,\x^{i-1}})}\left(1+\tilde{T}^\mu_{\eta,i}\right)^{im}\notag\\
&\leq\hat{c}'_i\left(1+\tilde{T}^\mu_{\eta,i}\right)^{im}\mu^{\rho^i_y-\eta\hat{C}'_i}
\end{align}
where $\hat{c}'_i:=2^n\hat{c} \ee^{\bar{T}_\eta\max_{z\in\Hn}r(z)}(1+\bar{T}_\eta)^m\hat{c}_{i-1}$ and $\hat{C}'_i\geq\hat{C}_{i-1}+(T^\mu_{\eta,i}-T_{i-1})\hat{C}+C\max_{z\in\Hn}f_{z,\x^{i-1}}$. As above, $\hat{C}'_i$ can be chosen uniformly in $\eta$ since $T^\mu_{\eta,i}\leq T_i+\eta C$ by Step 1, while $\hat{c}'_i$ may depend on $\eta$.

For $\tilde{\tau}(\xi^\mu_{\tilde{T}^\mu_{\eta,i}},\x^i)=\tau(\xi^\mu_{\tilde{T}^\mu_{\eta,i}},\x^i)\ln\frac{1}{\mu}$ and $\mu$ small enough, Step 1 implies
\begin{align}
|T^\mu_{\eta,i}+\tau(\xi^\mu_{\tilde{T}^\mu_{\eta,i}},\x^i)-T_i|\leq\eta C+\frac{\bar{T}_\eta}{\ln\frac{1}{\mu}}\leq2\eta C.
\end{align}
For $T_i<t<T_{i+1}$, we can now pick $\eta$ small enough such that $T_i+2\eta C<t<T_{i+1}-\eta C$, and hence
\begin{align}
\limsup_{\mu\to 0}T^\mu_{\eta,i}+\tau(\xi^\mu_{\tilde{T}^\mu_{\eta,i}},\x^i)<t<\liminf_{\mu\to0}T^\mu_{\eta,i+1}.
\end{align}
As in Corollary \ref{Corexp}, with the above bounds on $\xi^\mu_{\tilde{T}^\mu_{\eta,i}+\tilde{\tau}(\xi^\mu_{\tilde{T}^\mu_{\eta,i}},\x^i)}$, we derive
\begin{align}
\xi^\mu_{t\ln 1/\mu}(y)&\geq\check{c}\check{c}'_i\mu^{\min_{z\in\Hn}[\rho^i_z+\eta\check{C}'_i+|z-y|-(t-(T^\mu_{\eta,i}+\tau(\xi^\mu_{\tilde{T}^\mu_{\eta,i}},\x^i)))(f_{z,\x^i}-\eta\check{C})]}\notag\\
&\geq\check{c}\check{c}'_i\mu^{\min_{z\in\Hn}[\rho^i_z+|z-y|-(t-T_i)(f_{z,\x^i}-\eta\check{C})]+\eta(\check{C}'_i+2C\max_{z\in\Hn}(f_{z,\x^i}-\eta\check{C}))}\notag\\
&=\check{c}_i\mu^{\min_{z\in\Hn}[\rho^i_z+|z-y|-(t-T_i)(f_{z,\x^i}-\eta\check{C})]+\eta\check{C}_i},
\end{align}
defining $\check{c}_i:=\check{c}\check{c}'_i$ and $\check{C}_i:=\check{C}'_i+2C\max_{z\in\Hn}(f_{z,\x^i}-\eta\check{C})$.

Similarly, the upper bound is derived as
\begin{align}
\xi^\mu_{t\ln 1/\mu}&(y)\leq2^n\hat{c}\hat{c}'_i\mu^{\min_{z\in\Hn}[\rho^i_z-\eta\hat{C}'_i+|z-y|-(t-(T^\mu_{\eta,i}+\tau(\xi^\mu_{\tilde{T}^\mu_{\eta,i}},\x^i)))(f_{z,\x^i}+\eta\hat{C})]}\notag\\
&\qquad\cdot(1+\tilde{T}^\mu_{\eta,i})^{im}\left(1+\left(t\ln\frac{1}{\mu}-(\tilde{T}^\mu_{\eta,i}+\tilde{\tau}(\xi^\mu_{\tilde{T}^\mu_{\eta,i}},\x^i))\right)\right)^m\notag\\
&\leq 2^n\hat{c}\hat{c}'_i\mu^{\min_{z\in\Hn}[\rho^i_z+|z-y|-(t-T_i)(f_{z,\x^i}+\eta\hat{C})]-\eta(\hat{C}'_i+2C\max_{z\in\Hn}(f_{z,\x^i}+\eta\hat{C}))}\left(1+t\ln\frac{1}{\mu}\right)^{(i+1)m}\notag\\
&=\hat{c}_i\mu^{\min_{z\in\Hn}[\rho^i_z+|z-y|-(t-T_i)(f_{z,\x^i}+\eta\hat{C})]-\eta\hat{C}_i}\left(1+t\ln\frac{1}{\mu}\right)^{(i+1)m},
\end{align}
with $\hat{c}_i:=2^n\hat{c}\hat{c}'_i$ and $\hat{C}_i:=\hat{C}'_i+2C\max_{z\in\Hn}(f_{z,\x^i}+\eta\hat{C})$. This concludes the proof of (\ref{ClaimSize}).

Notice, that, although $\check{c}_i$ and $\hat{c}_i$ may vary for different $\eta$, $\check{C}_i$ and $\hat{C}_i$ can be chosen uniformly in $\eta$.

For every $x\in\x^i$, we obtain $\xi^\mu_{t\ln1/\mu}(x)\in[\bar{\xi}_{\x^i}(x)-\eta\bar{C},\bar{\xi}_{\x^i}(x)+\eta\bar{C}]$, as in Theorem \ref{Thmexp}.


\noindent\textit{Step 4:} Convergence for $T_i<t<T_{i+1}$.\\
We now want to prove the actual convergence. We already know that the resident types are staying close to their equilibrium between $T_i$ and $T_{i+1}$ and therefore mainly have to show that the population sizes of the non-resident types vanish as $\mu\to0$.

We claim that, for each $i\geq0$, $T_i<t<T_{i+1}$, and $y\in\Hn\backslash\x^i$,
\begin{align}
\min_{z\in\Hn}[\rho^i_z+|z-y|-(t-T_i)(f_{z,\x^i}+\eta\hat{C})]-\eta\hat{C}_i\geq\gamma,
\end{align}
for some $\gamma>0$ and all $\eta$ small enough, and hence
\begin{align}
0\leq\lim_{\mu\to0}\xi^\mu_{t\ln1/\mu}(y)\leq\lim_{\mu\to0}\hat{c}_i\mu^\gamma\left(1+t\ln\frac{1}{\mu}\right)^{(i+1)m}=0.
\end{align}

We distinguish several cases. If $z\in\x^i$, this implies $f_{z,\x^i}=0$, $\rho^i_z=0$, and $|z-y|\geq1$. Hence
\begin{align}
\rho^i_z+|z-y|-(t-T_i)(f_{z,\x^i}+\eta\hat{C})-\eta\hat{C}_i\geq 1-\eta((t-T_i)\hat{C}+\hat{C}_i).
\end{align}

If $z\in\Hn\backslash\x^i$ and $\rho^i_z=0$, this implies $f_{z,\x^i}<0$ and
\begin{align}
\rho^i_z+|z-y|-(t-T_i)(f_{z,\x^i}+\eta\hat{C})-\eta\hat{C}_i\geq -(t-T_i)f_{z,\x^i}-\eta((t-T_i)\hat{C}+\hat{C}_i).
\end{align}

If $z\in\Hn\backslash\x^i$, $\rho^i_z>0$, and $f_{z,\x^i}\leq0$, we get
\begin{align}
\rho^i_z+|z-y|-(t-T_i)(f_{z,\x^i}+\eta\hat{C})-\eta\hat{C}_i\geq\rho^i_z-\eta((t-T_i)\hat{C}+\hat{C}_i).
\end{align}

Since $\check{C}_i$ does not depend on $\eta$, all these expressions can be bounded from below by a positive constant $\gamma$ if $\eta$ is small enough.

Finally, if $z\in\Hn\backslash\x^i$, $\rho^i_z>0$, and $f_{z,\x^i}>0$, we obtain $t<T_{i+1}\leq\rho^i_z/f_{z,\x^i}+T_i$ and, for $\eta$ and $\gamma$ small enough, $t-T_i<(\rho^i_z-\eta\hat{C}_i-\gamma)/(f_{z,\x^i}+\eta\hat{C})$. Therefore,
\begin{align}
\rho^i_z+|z-y|-(t-T_i)(f_{z,\x^i}+\eta\hat{C})-\eta\hat{C}_i>\rho^i_z-\eta\hat{C}_i-(\rho^i_z-\eta\check{C}_i-\gamma)=\gamma.
\end{align}
This proves the claim, in particular in the case where $T_{i+1}=\infty$ and there is no $y\in\Hn$ such that $f_{y,\x^i}>0$.

Last, we consider the $x\in\x^i$. For every $\eta$ small enough, 
\begin{align}
\lim_{\mu\to0}\xi^\mu_{t\ln1/\mu}(x)\in[\bar{\xi}_{\x^i}(x)-\eta\bar{C},\bar{\xi}_{\x^i}(x)+\eta\bar{C}].
\end{align}
As a result, $\lim_{\mu\to0}\xi^\mu_{t\ln1/\mu}(x)=\bar{\xi}_{\x^i}(x)$ and
\begin{align}
\lim_{\mu\to0}\xi^\mu_{t\ln1/\mu}=\sum_{x\in\x^i}\delta_x\bar{\xi}_{\x^i}(x).
\end{align}
\end{proof}


\section{Special Case of Equal Competition}

In this section we turn to the proof of Theorem \ref{EqComp}, the special case of equal competition between types. We go through the proof of Theorem \ref{MainThm} to make changes where assumptions are no longer satisfied and check the identities for $x^i$ and $T_i$.

\begin{proof}[Theorem \ref{EqComp}]
%
Unfortunately, assumption (B$_\x$) is not satisfied since there are no constants $\theta_x$ such that $(\theta_x\alpha)_{x,y\in\x}$ is positive definite for $|\x|\geq 2$. To still be able to apply the results of Theorem \ref{MainThm}, we have to carefully go through all the points, where assumption (B$_\x$) was used.

In the proof of Theorem \ref{Thmexp}, this property is only used for the resident types $\x$. In the case where $\x$ consists of a single type, the positive definiteness is trivially satisfied since $\alpha>0$.

In the case of Theorem \ref{LVThm}, we have to argue differently in a few places. \cite{ChJaRa10} derive Proposition 1 from a more general theorem. If one adapts the proof of this theorem to our situation, one sees that assumption (B$_\x$) is first used to prove that there are only finitely many equilibrium points.
In our special case, we are only considering Lotka-Volterra systems involving the old resident type $x^{i-1}$ and the minimizing mutant $y^i_*=x^i$. An equilibrium point $\xi^*\in (\R_{\geq0})^{\{x^{i-1},x^i\}}$ has to satisfy
\begin{align}
&\xi^*(x^{i-1})=0\text{ or }r(x^{i-1})=\alpha(\xi^*(x^{i-1})+\xi^*(x^i)),\notag\\
\text{and }&\xi^*(x^i)=0\text{ or }r(x^i)=\alpha(\xi^*(x^{i-1})+\xi^*(x^i)).
\end{align}
Since $f_{x^i,x^{i-1}}>0$, we obtain $r(x^i)>r(x^{i-1})$ and there are only three equilibrium points, namely $(0,0)$, $(r(x^{i-1})/\alpha,0)$, and $(0,r(x^i)/\alpha)$.

Moreover, assumption (B$_\x$) is used to prove that the evolutionary stable state (if existent) is unique. An evolutionary stable state $\bar{\xi}\in (\R_{\geq0})^{\{x^{i-1},x^i\}}$ is characterised by
\begin{align}
\begin{cases}r(x^j)-\alpha(\bar{\xi}(x^{i-1})+\bar{\xi}(x^i))\leq0,\text{if }\bar{\xi}(x^j)=0,\\
r(x^j)-\alpha(\bar{\xi}(x^{i-1})+\bar{\xi}(x^i))=0,\text{if }\bar{\xi}(x^j)>0,\end{cases}
\end{align}
for $j\in\{i-1,i\}$. Since $f_{i,i+1}>0$, only the last of the three equilibrium points satisfies these assumptions,
\begin{align}
r(x^{i-1})-\alpha(\bar{\xi}(x^{i-1})+\bar{\xi}(x^i))=r(x^{i-1})-\alpha\left(0+\frac{r(x^i)}{\alpha}\right)=-f_{i,i-1}\leq0,\\
r(x^i)-\alpha(\bar{\xi}(x^{i-1})+\bar{\xi}(x^i))=r(x^i)-\alpha\left(0+\frac{r(x^i)}{\alpha}\right)=0.
\end{align}

Finally, in Lemma \ref{conttau}, we are again in the situation where $\x$ consists of only one type and hence the positive definiteness is trivial.

The only thing left is to show the identities for $x^i$ and $T_i$. We claim that, for $i\geq0$,
\begin{align}
\rho^{i+1}_y=\min_{z_{i+1}\in\Hn}\cdots\min_{z_1\in\Hn}
  \Bigg[&|y-z_{i+1}|+\sum\limits_{j=1}^{i}|z_{j+1}-z_j|+|z_1-x^0|\notag\\
&-f_{z_1,x^0}T_1-\sum_{j=1}^{i} f_{z_{j+1},x_j}(T_{j+1}-T_j)\Bigg].
\end{align}

From the initial condition we obtain $\rho^0_y=\min_{z\in\Hn}[\lambda_z+|z-y|]=|y-x^0|$. Hence,
\begin{align}
y^1_*=\argmin_{y\in\Hn:f_{y,x^0}>0}\frac{|y-x^0|}{f_{y,x^0}}
\end{align}
and
\begin{align}
T_1=\min_{\substack{y\in\Hn:\\f_{y,x^0}>0}}\frac{|y-x^0|}{f_{y,x^0}}.
\end{align}
Since $f_{y^1_*,x^0}=r(y^1_*)-r(x^0)>0$, the new equilibrium is monomorphic of type $x^1=y^1_*$ and $T_1=|x^1-x^0|/f_{1,0}$. Moreover,
\begin{align}
\rho^1_y=\min_{z\in\Hn}[\rho^0_z+|z-y|-T_1f_{z,x^0}]=\min_{z\in\Hn}\left[|y-z|+|z-x^0|-f_{z,x^0}T_1\right].
\end{align}

Assume that $x^i$, $T_i$, and $\rho^i_y$ are of the proposed form. Then there is a unique
\begin{align}
&x^{i+1}=y^{i+1}_*=\argmin_{y\in\Hn:f_{y,x^i}>0}\frac{\rho^i_y}{f_{y,x^i}}\notag\\
&=\argmin_{y\in\Hn:f_{y,x^i}>0}\frac{\min_{z_i\in\Hn}\left[|y-z_i|+\rho^{i-1}_{z_i}-f_{z_i,x^{i-1}}(T_i-T_{i-1})\right]}{f_{y,x^i}}\notag\\
&=\argmin_{y\in\Hn:f_{y,x^i}>0}\min_{z_i\in\Hn}F(y,z_i),
\end{align}
where the last equality serves as the definition of the function $F:\Hn\times\Hn\to\R_+$.

Assume that the minimum over $z_i$ is only realised by some $\bar{z}\neq y^{i+1}_*$, i.e.
\begin{align}
\min_{\substack{y\in\Hn\\f_{y,x^i}>0}}\min_{z_i\in\Hn}F(y,z_i)=\min_{z_i\in\Hn}F(y^{i+1}_*,z_i)=F(y^{i+1}_*,\bar{z})<F(y^{i+1}_*,y^{i+1}_*).
\end{align}
Looking back at the definition of $F$ and using that
\begin{align}
\rho^{i-1}_{y^{i+1}_*}&=\min_{z\in\Hn}[\rho^{i-2}_z+|z-y^{i+1}_*|-(T_{i-1}-T_{i-2})f_{z,\x^{i-2}}]\notag\\
&\leq\min_{z\in\Hn}[\rho^{i-2}_z+|z-\bar{z}|-(T_{i-1}-T_{i-2})f_{z,\x^{i-2}}]+|\bar{z}-y^{i+1}_*|\notag\\
&=\rho^{i-1}_{\bar{z}}+|y^{i+1}_*-\bar{z}|,
\end{align}
this yields
\begin{align}
0\leq |y^{i+1}_*-\bar{z}|+\rho^{i-1}_{\bar{z}}-\rho^{i-1}_{y^{i+1}_*}<(f_{\bar{z},x^{i-1}}-f_{y^{i+1}_*,x^{i-1}})(T_i-T_{i-1})
\end{align}
and, since $T_i>T_{i-1}$, we obtain $f_{\bar{z},x^{i-1}}>f_{y^{i+1}_*,x^{i-1}}>0$. But this would imply
\begin{align}
\min_{z_i\in\Hn}F(\bar{z},z_i)\leq F(\bar{z},\bar{z})<F(y^{i+1}_*,\bar{z})=\min_{\substack{y\in\Hn\\f_{y,x^i}>0}}\min_{z_i\in\Hn}F(y,z_i),
\end{align}
which is a contradiction. Hence, $\bar{z}$ can be chosen equal to $y^{i+1}_*$.

Repeating the previous argument shows that the minimum oder $z_1,...,z_{i+1}$ is achieved at $z_1=...=z_{i+1}=y$ and hence
\begin{align}
x^{i+1}&=\argmin_{y\in\Hn:f_{y,x^i}>0}\frac{\rho^{i-1}_y-f_{y,x^{i-1}}(Ti-T_{i-1})}{f_{y,x^i}}=\dots\notag\\
&=\argmin_{y\in\Hn:f_{y,x^i}>0}\frac{|y-x^0|-f_{y,x^0}\frac{|x^1-x^0|}{f_{1,0}}-\sum_{j=1}^{i-1}f_{y,x^j}(T_{j+1}-T_j)}{f_{y,x^i}}\notag\\
&=\argmin_{y\in\Hn:f_{y,x^i}>0}\frac{|y-x^0|}{f_{y,x^i}}-\frac{f_{y,x^{i-1}}}{f_{y,x^i}}T_i-\sum_{j=1}^{i-1}\frac{f_{y,x^{j-1}}-f_{y,x^j}}{f_{y,x^i}}T_j\notag\\
&=\argmin_{y\in\Hn:f_{y,x^i}>0}\frac{|y-x^0|}{f_{y,x^i}}-\frac{f_{y,x^{i-1}}(|x^i-x^0|-|x^{i-1}-x^0|)}{f_{y,x^i}f_{i,i-1}}\notag\\
&\qquad\qquad\qquad\qquad-\sum_{j=1}^{i-1}\frac{|x^j-x^0|-|x^{j-1}-x^0|}{f_{j,j-1}}\frac{f_{y,x^{j-1}}-f_{y,x^j}}{f_{y,x^i}}\notag\\
&=\argmin_{y\in\Hn:f_{y,x^i}>0}\frac{|y-x^0|}{f_{y,x^i}}-(|x^i-x^0|-|x^{i-1}-x^0|)\left(\frac{1}{f_{i,i-1}}+\frac{1}{f_{y,x^i}}\right)\notag\\
&\qquad\qquad\qquad\qquad-\frac{|x^{i-1}-x^0|-|x^0-x^0|}{f_{y,x^i}}\notag\\
&=\argmin_{y\in\Hn:f_{y,x^i}>0}\frac{|y-x^0|-|x^i-x^0|}{f_{y,x^i}}-T_i,
\end{align}
where we use (\ref{trans}) several times. Analogously,
\begin{align}
T_{i+1}&=T_i+\min_{\substack{y\in\Hn:\\f_{y,x^i}>0}}\frac{\rho^i_y}{f_{y,x^i}}=T_i+\left(\frac{|x^{i+1}-x^0|-|x^i-x^0|}{f_{i+1,i}}-T_i\right)=\frac{|x^{i+1}-x^0|-|x^i-x^0|}{f_{i+1,i}}.
\end{align}

Finally,
\begin{align}
\rho^{i+1}_y=\min_{z_{i+1}\in\Hn}[\rho^i_{z_{i+1}}+|z_{i+1}-y|-(T_{i+1}-T_i)f_{z_{i+1},x^i}],
\end{align}
which is of the desired form. This proves the claim and hence the theorem.
\end{proof}


\section{A First Look at Limited Range of Mutation}

In this section we present the proof of Theorem \ref{AW}, where $\ell=1$, and take a first look at the intermediate cases of $1<\ell<n$.

\subsection{Proof for the case $\ell=1$}

We again go over the previous proofs and make alterations where necessary.
\begin{proof}[Theorem \ref{AW}]
We only consider the first invasion step. We can assume that $\eta<\bar{\xi}$. Consequently, up to time $\tilde{T}^\mu_{\eta,1}\land\inf\{t\geq0:\exists\ z\in\Hn, |z-x^0|>1,\xi^\mu_t(z)\geq\bar{\xi}\mu\}$, the neighbours of type $x^0$ are the only active mutants. As before,
\begin{align}
\xi^\mu_t(x^0)\in [\bar{\xi}_{x^0}(x^0)-\eta\bar{C},\bar{\xi}_{x^0}(x^0)+\eta\bar{C}].
\end{align}
Moreover, as in (\ref{upbnd}) and (\ref{lowbnd}), we obtain
\begin{align}
[f_{x^0,x^0}-\eta\check{C}]\xi^\mu_t(x^0)\leq\tfrac{d}{dt}\xi^\mu_t(x^0)\leq[f_{x^0,x^0}+\eta\hat{C}]\xi^\mu_t(x^0),
\end{align}
and with $f_{x^0,x^0}=0$, $c:=\bar{\xi}_{x^0}(x^0)-\bar{c}\bar{\xi}$, and $C:=\bar{\xi}_{x^0}(x^0)+\bar{c}\bar{\xi}$,
\begin{align}
c\ee^{-t\eta\check{C}}\leq\xi^\mu_t(x^0)\leq C\ee^{t\eta\hat{C}}.
\end{align}

Considering the neighbours $y\sim x^0$ of the resident type, we derive
\begin{align}
[f_{y,x^0}-\eta\check{C}]\xi^\mu_t(y)+\mu\tilde{c}\xi^\mu_t(x^0)\leq\tfrac{d}{dt}\xi^\mu_t(y)\leq[f_{y,x^0}+\eta\hat{C}]\xi^\mu_t(y)+\mu\tilde{C}\xi^\mu_t(x^0),
\end{align}
and hence 
%
the upper bound,
\begin{align}
\xi^\mu_t(y)&\leq \ee^{t(f_{y,x^0}+\eta\hat{C})}C_y\mu^{\lambda_y}+\mu\tilde{C}C\int_0^t\ee^{s\eta\hat{C}}\ee^{(t-s)(f_{y,x^0}+\eta\hat{C})}ds\notag\\
&\leq\mu \ee^{t(f_{y,x^0}+\eta\hat{C})}\left(C_y\mu^{\lambda_y-1}+\tilde{C}C\int_0^t\ee^{-sf_{y,x^0}}ds\right)\notag\\
&\leq \hat{c}'\mu \ee^{t\eta\hat{C}}\left((1+t)\ee^{tf_{y,x^0}}+1\right),
\end{align}
for some $\hat{c}'<\infty$, uniformly in $y\sim x^0$, $\eta<\bar{\xi}$, and $\mu$.

A similar lower bound can be shown and, on the $\ln1/\mu$-time scale, we obtain
\begin{align}
\check{c}'\mu^{((1-tf_{y,x^0})\land1)+t\eta\check{C}}\leq\xi^\mu_{t\ln\frac{1}{\mu}}(y)\leq\hat{c}'\mu^{((1-tf_{y,x^0})\land1)-t\eta\hat{C}}\left(1+t\ln\frac{1}{\mu}\right).
\end{align}

Using this bound, all types $z$ such that $|z-x^0|=2$ can be bounded from above using the same type of calculation to derive
\begin{align}
\xi^\mu_{t\ln\frac{1}{\mu}}(z)\leq C\mu^{2-t\eta\hat{C}}\left(\left(1+t\ln\frac{1}{\mu}\right)^2\mu^{-t\max_{y\sim x^0}f_{y,x^0}}+1\right).
\end{align}
Hence, for $\eta$ small enough, $\tilde{T}^\mu_{\eta,1}\approx\inf\{t\geq0:\exists\ z\in\Hn, |z-x^0|>1,\xi^\mu_t(z)\geq\bar{\xi}\mu\}$.

As in Corollary \ref{Corexp}, we can now argue that
\begin{align}
\min_{\substack{y\sim x^0\\f_{y,x^0}>0}}\frac{1}{f_{y,x^0}+\eta\check{C}}\leq \liminf_{\mu\to0}T^\mu_{\eta,1}\leq\limsup_{\mu\to0}T^\mu_{\eta,1}\leq\min_{\substack{y\sim x^0\\f_{y,x^0}>0}}\frac{1}{f_{y,x^0}-\eta\check{C}}.
\end{align}
The first mutant $y^1_*$ to reach the $\eta$-level is the neighbour of $x^0$ minimising $1/f_{y,x^0}$ (given $f_{y,x^0}>0$), hence maximising $r(y)$, which is unique (or else we set $I:=i$ and terminate the procedure). This yields $T^\mu_{\eta,1}\approx T_1=1/f_{y^1_*,x^0}$.

The Lotka-Volterra phase can be analysed just as before. Since $y^1_*$ satisfies $r(y^1_*)>r(x^0)$, the new equilibrium has $x^1=y^1_*$ as the only resident type.

Since, for every other $y\sim x^0$, $r(y)<r(x^1)$, these types always stay unfit, do not foster mutants above the threshold, and we do not need to consider them any further.

During the Lotka-Volterra phase, once the $\xi^\mu_t(z)$, $z\sim x^1$ have surpassed $\bar{\xi}\mu$, they start to grow. However, since the duration of the Lotka-Volterra phase can be bounded uniformly as before, this only results in mutant populations of order $\mu^{1}$, which fits the initial conditions for the next invasion step.
\end{proof}


\subsection{The intermediate cases}

For now, we stick with the assumption of constant competition. In the case of $\ell\geq n$, arbitrarily large steps can be taken. In particular, arbitrarily large valleys in the fitness landscape (defined by $r$) can be crossed. A (strict) global fitness maximum is reached eventually and is the only stable point. If $\ell=1$, the limiting walk always jumps to the fittest nearest neighbour and (strict) local fitness maxima are stable points. In both cases, the microscopic types do not have to be tracked to characterise the jump process. The next step is determined only by the previous and possibly the initial resident type.

The cases $2\leq\ell\leq n-1$ interpolate between the two extreme scenarios. To study accessibility of different types, we again need to keep track of the microscopic populations. To this extent, we define some new quantities.
\begin{definition}
The \textit{first appearance time} of a type $y$ (on the $\ln 1/\mu$-time scale) is denoted by
\begin{align}
\tau^\mu_y:=\inf\{s\geq 0:\xi^\mu_{s\ln\frac{1}{\mu}}(y)>0\}.
\end{align}

The $\mu$-power the population size of type $y$ would have at time $t\ln1/\mu$ due to its own growth rate (neglecting mutation from neighbours after $\tau^\mu_y$) is
\begin{align}
\lambda_t(y):=\1_{t\geq\tau^\mu_y}
     \Big(\underbrace{\ell\land|y-x^0|}_\text{initial size}-\sum_{i=0}^\infty \underbrace{f_{y,x^i}(t\land T^\mu_{\eta,i+1}-\tau^\mu_y\lor T^\mu_{\eta,i})_+}_{\substack{\text{growth between}\\ i^\text{th}\text{ and }(i+1)^\text{st}\text{ invasion}}}\Big)+\1_{t<\tau^\mu_y}\infty,
\end{align}
where $x^i$ and $T^\mu_{\eta,i}$ are just as before.

 All types under the \textit{mutational influence} of type $y$ are denoted by
 \begin{align}
 \Lambda_t(y):=\{z\in\Hn: |z-y|+\lambda_t(y)\leq\ell\}
 \end{align}
 and $\Lambda_t:=\bigcup_{y\in\Hn}\Lambda_t(y)$.
\end{definition}

Since we are assuming constant competition, the population sizes of the different types are approximated by
\begin{align}
\xi^\mu_{t\ln\frac{1}{\mu}}(y)\approx\1_{y\in \Lambda_t}\mu^{\min_{z\in \Lambda_t}[|y-z|+\lambda_t(z)]},
\end{align}
where we drop multiplicative constants and all terms involving $\eta$. Figure 3 visualises the interplay of $\lambda_t(y)$, $\xi^\mu_{t\ln1/\mu}(y)$, and the sets $\Lambda_t(y)$ for an easy example.

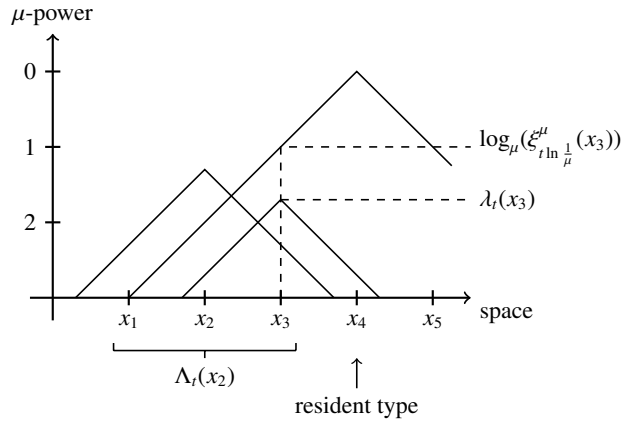
\begin{figure}[h]
\centering
 \scriptsize
 \begin{tikzpicture}
 \draw[thick,->](-0.3,0)--(5.5,0) node[anchor=north west]{space};
 \draw[thick,->](0,-0.3)--(0,3.5) node[anchor=south]{$\mu$-power};
 \draw[thick](1,0.1)--(1,-0.1) node[anchor=north]{$x_1$};
 \draw[thick](2,0.1)--(2,-0.1) node[anchor=north]{$x_2$};
 \draw[thick](3,0.1)--(3,-0.1) node[anchor=north]{$x_3$};
 \draw[thick](4,0.1)--(4,-0.1) node[anchor=north]{$x_4$};
 \draw[thick](5,0.1)--(5,-0.1) node[anchor=north]{$x_5$};
 \draw[thick](0.1,1)--(-0.1,1) node[anchor=east]{$2$};
 \draw[thick](0.1,2)--(-0.1,2) node[anchor=east]{$1$};
 \draw[thick](0.1,3)--(-0.1,3) node[anchor=east]{$0$};
 \draw[semithick](0.3,0)--(2,1.7)--(3.7,0);
 \draw[semithick](1.7,0)--(3,1.3)--(4.3,0);
 \draw[semithick]((1,0)--(4,3)--(5.25,1.75);
 \draw[semithick,dashed](3,0)--(3,2)--(5.5,2)node[anchor=west]{$\log_\mu(\xi^\mu_{t\ln\frac{1}{\mu}}(x_3))$};
 \draw[semithick,dashed](3,1.3)--(5.5,1.3)node[anchor=west]{$\lambda_t(x_3)$};
 \draw[semithick](0.8,-0.7)--(3.2,-0.7) (0.8,-0.7)--(0.8,-0.6) (3.2,-0.7)--(3.2,-0.6) ((2,-0.7)--(2,-0.8)node[anchor=north]{$\Lambda_t(x_2)$};
 \draw[semithick,<-](4,-0.8)--(4,-1.2)node[anchor=north]{resident type};
 \end{tikzpicture}
 \caption{Example for the case $\ell=3$. The mutational influence of $x_2$ reaches $x_1$ and $x_3$. The population size of $x_3$ is not determined by its own growth rate but by mutants from the resident type $x_4$.}
 \end{figure}\pagebreak
 
 It is not easy to make general statements about the evolution of this intermediate model. However, we state some first results on the accessibility of types.
\begin{definition}
A type $y\in\Hn$ is called \textit{accessible} if $y\in\Lambda_\infty:=\bigcup_{t\geq0}\Lambda_t$.
\end{definition}
\begin{remark}
This is equivalent to $\tau^\mu_y<\infty$.
\end{remark}

Since resident types can only produce mutants in a radius of $\ell$, in order to be accessible, a type has to be reached on a path with types of increasing fitness and at most distance $\ell$. Figure 4 gives an example for such a path.
\begin{lemma}\label{accessibility}
A necessary condition for a type $y$ to be accessible is the existence of a path \linebreak $(y_0=x^0,y_1,...,y_m=y)$ and indices $i_0=0<i_1<...<i_k=m$, such that
\begin{align}
\forall\ 1\leq j\leq k: &\ |i_j-i_{j-1}|\leq \ell,\\
\forall\ 1\leq j< k: &\ f_{y_{i_j},y_{i_{j-1}}}>0,\\
\forall\ i_{j-1}<i<i_j: &\ f_{y_{i_{j-1}},y_i}>0.
\end{align}
\end{lemma}

\begin{figure}[h]
\centering
 \scriptsize
 \begin{tikzpicture}
 \draw[thick,->] (-1.5,0)--(9.5,0) node[anchor=north west]{$\Hn$};
 \draw[thick,->](-1,-0.5)--(-1,2.5) node[anchor=east]{$r(x)$};
 \draw[thick] (0,0.1)--(0,-0.1) node[anchor=north]{$y_0=x^0$};
     \fill[black] (0,0.3) circle (0.05) node[anchor=south]{$i_0$};
     \draw[dashed] (0,0.3)--(1,0.3);
 \draw[thick] (1,0.1)--(1,-0.1) node[anchor=north]{$y_1$};
     \fill[black] (1,0.5) circle (0.05) node[anchor=south]{$i_1$};
     \draw[dashed] (1,0.5)--(4,0.5);
 \draw[thick] (2,0.1)--(2,-0.1) node[anchor=north]{$y_2$};
     \fill[black] (2,0.2) circle (0.05);
 \draw[thick] (3,0.1)--(3,-0.1) node[anchor=north]{$y_3$};
     \fill[black] (3,0.4) circle (0.05);
 \draw[thick] (4,0.1)--(4,-0.1) node[anchor=north]{$y_4$};
     \fill[black] (4,1) circle (0.05) node[anchor=south]{$i_2$};
     \draw[dashed] (4,1)--(6,1);
 \draw[thick] (5,0.1)--(5,-0.1) node[anchor=north]{$y_5$};
     \fill[black] (5,0.6) circle (0.05);
 \draw[thick] (6,0.1)--(6,-0.1) node[anchor=north]{$y_6$};
     \fill[black] (6,1.4) circle (0.05) node[anchor=south]{$i_3$};
     \draw[dashed] (6,1.4)--(7,1.4);
 \draw[thick] (7,0.1)--(7,-0.1) node[anchor=north]{$y_7$};
     \fill[black] (7,1.6) circle (0.05) node[anchor=south]{$i_4$};
     \draw[dashed] (7,1.6)--(9,1.6);
 \draw[thick] (8,0.1)--(8,-0.1) node[anchor=north]{$y_8$};
     \fill[black] (8,0.3) circle (0.05);
 \draw[thick] (9,0.1)--(9,-0.1) node[anchor=north]{$y_9=y$};
     \fill[black] (9,1.3) circle (0.05) node[anchor=south west]{$i_5$};
 \draw[thick](0.9,1.5)--(0.9,1.6)--(4.1,1.6)--(4.1,1.5) (2.5,1.6)--(2.5,1.7) node[anchor=south]{$\leq\ell=3$};
 \end{tikzpicture}
 \caption{A possible path to access $y$, for $\ell=3$.}
 \end{figure}
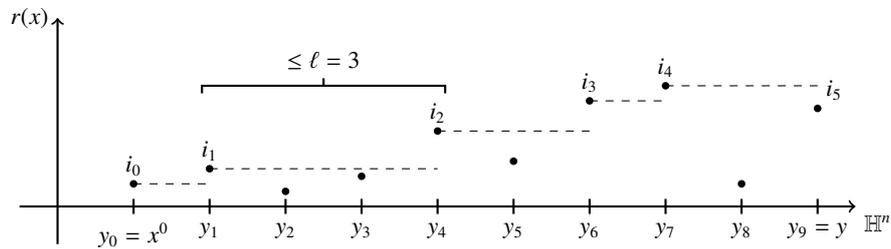

\begin{proof}
Assume that $y\neq x^0$. If $y\in\Lambda_0(x^0)$, this implies $|y-x^0|\leq \ell$. Hence we can choose any shortest path from $x^0$ to $y$  and pick the indices $i_j$ such that the conditions are satisfied.

If $y$ is accessible but $y\notin\Lambda_0(x^0)$, then $\tau^\mu_y>0$. There is at least one $z\neq y$ such that $y\in\Lambda_{\tau^\mu_y}(z)$. We choose such a $z$ for which the rate $r(z)$ is maximal. Consequently, $\tau^\mu_z<\tau^\mu_y$ and $\xi^\mu_{\tau^\mu_y\ln\frac{1}{\mu}}(z)\approx\mu^{\lambda_{\tau^\mu_y}(z)}$ (else $z$ would just grow due to mutants from a fitter type, which would imply that $z$ was not chosen such that the rate $r(z)$ is maximal). Any direct path from $z$ to $y$ now only goes through types that are unfit in comparison to $z$. We set $y_{i_k}:=z$.

We can now iterate this procedure with $z$ replacing $y$. In addition, we know that, for the $z'\neq z$ such that $z\in\Lambda_{\tau^\mu_z}(z')$ and $r(z')$ is maximised, $r(z)>r(z')$ (else, we would obtain $\Lambda_t(z)\subset\Lambda_t(z')$, for all $t\geq 0$, and $z$ would not have been chosen maximising $r(z)$). We set $y_{i_{k-1}}:=z'$ and continue until we reach $x^0$.
\end{proof}

\begin{remark}
The condition in Lemma \ref{accessibility} is not sufficient. Even if such a path exists, there might be a type $z$ that is reached before $y_{i_j}$ such that $r(z)>r(y_{i_j})$. In this case the population of $y_{i_j}$ is not fit to grow and might never reach the necessary size to induce mutants of type $y_{i_{j+1}}$.
\end{remark}

As a Corollary, we can consider the non-crossing of fitness valleys. Figure 5 gives the example of a non-accessible type, surrounded by a fitness valley.

\begin{corollary}\label{non-crossing}
If a type $y$ is surrounded by a fitness valley of width at least $\ell+1$, i.e.\ for all paths \linebreak $(y_0=x_0,y_1,...,y_m=y)$ there exists an $i\leq m-(\ell+1)$ such that $f_{y_i,y_j}>0,\forall\ i<j<m$, it is non-accessible.
\end{corollary}

\begin{figure}[h]
\centering
 \scriptsize
 \begin{tikzpicture}
 \draw[thick,->] (-1.5,0)--(9.5,0) node[anchor=north west]{$\Hn$};
 \draw[thick,->](-1,-0.5)--(-1,2.5) node[anchor=east]{$r(x)$};
 \draw[thick] (0,0.1)--(0,-0.1) node[anchor=north]{$y_{i_1}$};
     \fill[black] (0,1.7) circle (0.05) node[anchor=south]{$i_1$};
     \draw[dashed] (0,1.7)--(4,1.7);
 \draw[thick] (1,0.1)--(1,-0.1) node[anchor=north]{$ $};
     \fill[black] (1,0.9) circle (0.05);
 \draw[thick] (2,0.1)--(2,-0.1) node[anchor=north]{$ $};
     \fill[black] (2,0.3) circle (0.05);
 \draw[thick] (3,0.1)--(3,-0.1) node[anchor=north]{$ $};
     \fill[black] (3,1.2) circle (0.05);
 \draw[thick] (4,0.1)--(4,-0.1) node[anchor=north]{$y$};
     \fill[black] (4,1.5) circle (0.05);
 \draw[thick] (5,0.1)--(5,-0.1) node[anchor=north]{$ $};
     \fill[black] (5,0.6) circle (0.05);
 \draw[thick] (6,0.1)--(6,-0.1) node[anchor=north]{$ $};
     \fill[black] (6,1) circle (0.05);
 \draw[thick] (7,0.1)--(7,-0.1) node[anchor=north]{$ $};
     \fill[black] (7,1.2) circle (0.05);
 \draw[thick] (8,0.1)--(8,-0.1) node[anchor=north]{$ $};
     \fill[black] (8,0.4) circle (0.05);
 \draw[thick] (9,0.1)--(9,-0.1) node[anchor=north]{$y_{i_2}$};
     \fill[black] (9,1.4) circle (0.05) node[anchor=south]{$i_2$};
     \draw[dashed] (9,1.4)--(4,1.4);
 \draw[thick](-0.1,2.1)--(-0.1,2.2)--(4.1,2.2)--(4.1,2.1) (2,2.2)--(2,2.3) node[anchor=south]{$>\ell=3$};
 \end{tikzpicture}
 \caption{Due to the high fitness of $y_{i_1}$ and $y_{i_2}$, $y$ is not accessible for $\ell=3$.}
 \end{figure}
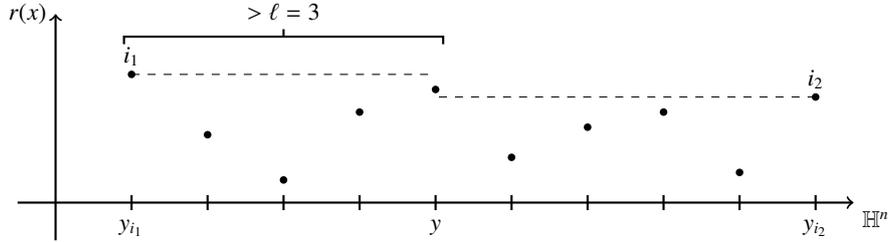

\begin{proof}
The claim follows directly from Lemma \ref{accessibility} since in this case the necessary path cannot exist.
\end{proof}
 
As a result, at least in the matter of crossing fitness valleys, the intermediate cases interpolate between the extreme cases.

However, as in the case of $\ell=n$, it is still possible to take arbitrarily large steps in the macroscopic process or the limiting jump process, respectively. If there was a series of types with distance smaller than $\ell+1$ and fast increasing rate $r$, then each population could be overtaken by its faster growing mutants before it reaches the macroscopic level of $\mu^0$.

Overall, the microscopic types play an important role in defining the limiting process.

\bibliographystyle{abbrv}

\begin{thebibliography}{10}

\bibitem{BaBo18}
M.~Baar and A.~Bovier.
\newblock The polymorphic evolution sequence for populations with phenotypic
  plasticity.
\newblock {\em Electron. J. Probab.}, 23(72):1--27, 2018.

\bibitem{BaBoCh17}
M.~Baar, A.~Bovier, and N.~Champagnat.
\newblock From stochastic, individual-based models to the canonical equation of
  adaptive dynamics in one step.
\newblock {\em Ann. Appl. Probab.}, 27(2):1093--1170, 2017.

\bibitem{BaPo05}
N.~H. Barton and J.~Polechov{\'a}.
\newblock The limitations of adaptive dynamics as a model of evolution.
\newblock {\em J. Evol. Biol.}, 18(5):1186--1190, 2005.

\bibitem{BeBruShi16}
J.~Berestycki, E.~Brunet, and Z.~Shi.
\newblock The number of accessible paths in the hypercube.
\newblock {\em Bernoulli}, 22(2):653--680, 2016.

\bibitem{BeBruShi17}
J.~Berestycki, E.~Brunet, and Z.~Shi.
\newblock Accessibility percolation with backsteps.
\newblock {\em ALEA Lat. Am. J. Probab. Math. Stat.}, 14(1):45--62, 2017.

\bibitem{BoPa97}
B.~Bolker and S.~W. Pacala.
\newblock Using moment equations to understand stochastically driven spatial
  pattern formation in ecological systems.
\newblock {\em Theor. Popul. Biol.}, 52(3):179--197, 1997.

\bibitem{BoPa99}
B.~M. Bolker and S.~W. Pacala.
\newblock Spatial moment equations for plant competition: understanding spatial
  strategies and the advantages of short dispersal.
\newblock {\em Am. Nat.}, 153(6):575--602, 1999.

\bibitem{BoCoNeu18}
A.~Bovier, L.~Coquille, and R.~Neukirch.
\newblock The recovery of a recessive allele in a {M}endelian diploid model.
\newblock {\em J. Math. Biol.}, 77(4):971--1033, 2018.

\bibitem{BoCoSm18}
A.~Bovier, L.~Coquille, and C.~Smadi.
\newblock Crossing a fitness valley as a metastable transition in a stochastic
  population model.
\newblock {\em Ann. Appl. Probab.}, online first, 2019.

\bibitem{BoWa13}
A.~Bovier and S.-D. Wang.
\newblock Trait substitution trees on two time scales analysis.
\newblock {\em Markov Process. Relat. Fields}, 19(4):607--642, 2013.

\bibitem{Cha06}
N.~Champagnat.
\newblock A microscopic interpretation for adaptive dynamics trait substitution
  sequence models.
\newblock {\em Stoch. Process. Appl.}, 116(8):1127--1160, 2006.

\bibitem{ChFeMe08}
N.~Champagnat, R.~Ferri\`ere, and S.~M\'{e}l\'{e}ard.
\newblock From individual stochastic processes to macroscopic models in
  adaptive evolution.
\newblock {\em Stoch. Models}, 24(suppl. 1):2--44, 2008.

\bibitem{ChJaRa10}
N.~Champagnat, P.-E. Jabin, and G.~Raoul.
\newblock Convergence to equilibrium in competitive {L}otka-{V}olterra and
  chemostat systems.
\newblock {\em C. R. Math. Acad. Sci. Paris}, 348(23-24):1267--1272, 2010.

\bibitem{ChMe07}
N.~Champagnat and S.~M\'{e}l\'{e}ard.
\newblock Invasion and adaptive evolution for individual-based spatially
  structured populations.
\newblock {\em J. Math. Biol.}, 55(2):147--188, 2007.

\bibitem{ChMe11}
N.~Champagnat and S.~M\'{e}l\'{e}ard.
\newblock Polymorphic evolution sequence and evolutionary branching.
\newblock {\em Probab. Theory Relat. Fields}, 151(1-2):45--94, 2011.

\bibitem{CoMeMe13}
P.~Collet, S.~M\'{e}l\'{e}ard, and J.~A.~J. Metz.
\newblock A rigorous model study of the adaptive dynamics of {M}endelian
  diploids.
\newblock {\em J. Math. Biol.}, 67(3):569--607, 2013.

\bibitem{DiLa96}
U.~Dieckmann and R.~Law.
\newblock The dynamical theory of coevolution: a derivation from stochastic
  ecological processes.
\newblock {\em J. Math. Biol.}, 34(5-6):579--612, 1996.

\bibitem{DiLa00}
U.~Dieckmann and R.~Law.
\newblock Moment approximations of individual-based models.
\newblock In {\em The geometry of ecological interactions: simplifying spatial
  complexity}, pages 252--270. Camb. Univ. Press, 2000.

\bibitem{DuMay11}
R.~Durrett and J.~Mayberry.
\newblock Traveling waves of selective sweeps.
\newblock {\em Ann. Appl. Probab.}, 21(2):699--744, 2011.

\bibitem{EtKu86}
S.~N. Ethier and T.~G. Kurtz.
\newblock {\em Markov processes}.
\newblock Wiley Ser. in Probab. and Math. Stat. John Wiley \& Sons, Inc., New
  York, 1986.

\bibitem{FoMe04}
N.~Fournier and S.~M\'{e}l\'{e}ard.
\newblock A microscopic probabilistic description of a locally regulated
  population and macroscopic approximations.
\newblock {\em Ann. Appl. Probab.}, 14(4):1880--1919, 2004.

\bibitem{HofSig98}
J.~Hofbauer and K.~Sigmund.
\newblock {\em Evolutionary games and population dynamics}.
\newblock Camb. Univ. Press, 1998.

\bibitem{Jain07}
K.~Jain.
\newblock Evolutionary dynamics of the most populated genotype on rugged
  fitness landscapes.
\newblock {\em Phys. rev. E, Stat., nonlinear, and soft matter phys.}, 76 3 Pt
  1:031922, 2007.

\bibitem{JainKrug05}
K.~Jain and J.~Krug.
\newblock Evolutionary trajectories in rugged fitness landscapes.
\newblock {\em J Stat Mech: Theory and Exp.}, 2005(04):P04008, apr 2005.

\bibitem{JainKrug07}
K.~Jain and J.~Krug.
\newblock Deterministic and stochastic regimes of asexual evolution on rugged
  fitness landscapes.
\newblock {\em Genetics}, 175:1275--88, 03 2007.

\bibitem{KaLe87}
S.~Kauffman and S.~Levin.
\newblock Towards a general theory of adaptive walks on rugged landscapes.
\newblock {\em J. Theor. Biol.}, 128(1):11--45, 1987.

\bibitem{Kau92}
S.~A. Kauffman.
\newblock The origins of order: Self-organization and selection in evolution.
\newblock In {\em Spin glasses and biology}, pages 61--100. World Scientific,
  1992.

\bibitem{KarlKrug03}
J.~Krug and C.~Karl.
\newblock Punctuated evolution for the quasispecies model.
\newblock {\em Physica A: Stat. Mech. Appl.}, 318(1):137 -- 143, 2003.
\newblock STATPHYS - Kolkata IV.

\bibitem{Lem16}
H.~Leman.
\newblock Convergence of an infinite dimensional stochastic process to a
  spatially structured trait substitution sequence.
\newblock {\em Stoch. Partial Differ. Equ. Anal. Comput.}, 4(4):791--826, 2016.

\bibitem{MaySmith62}
J.~Maynard~Smith.
\newblock {\em The scientist speculates: An anthology of partly-baked ideas.}
\newblock Basic Books, 1962.

\bibitem{MaySmith70}
J.~Maynard~Smith.
\newblock Natural selection and the concept of a protein space.
\newblock {\em Nature}, 225:563--564, 1970.

\bibitem{MeGeMe96}
J.~A.~J. Metz, S.~A.~H. Geritz, G.~Mesz\'{e}na, F.~J.~A. Jacobs, and J.~S. van
  Heerwaarden.
\newblock Adaptive dynamics, a geometrical study of the consequences of nearly
  faithful reproduction.
\newblock In {\em Stochastic and spatial structures of dynamical systems
  ({A}msterdam, 1995)}, volume~45 of {\em Konink. Nederl. Akad. Wetensch. Verh.
  Afd. Natuurk. Eerste Reeks}, pages 183--231. North-Holland, Amsterdam, 1996.

\bibitem{KrugNeid2011}
J.~Neidhart and J.~Krug.
\newblock Adaptive walks and extreme value theory.
\newblock {\em Phys. Rev. Lett.}, 107:178102, Oct 2011.

\bibitem{NeuBo17}
R.~Neukirch and A.~Bovier.
\newblock Survival of a recessive allele in a {M}endelian diploid model.
\newblock {\em J. Math. Biol.}, 75(1):145--198, 2017.

\bibitem{NoKr15}
S.~Nowak and J.~Krug.
\newblock Analysis of adaptive walks on {NK} fitness landscapes with different
  interaction schemes.
\newblock {\em J. Stat. Mech. Theory Exp.}, (6):P06014, 27, 2015.

\bibitem{Orr03}
H.~A. Orr.
\newblock A minimum on the mean number of steps taken in adaptive walks.
\newblock {\em J. Theor. Biol.}, 220(2):241--247, 2003.

\bibitem{SchKr14}
B.~Schmiegelt and J.~Krug.
\newblock Evolutionary accessibility of modular fitness landscapes.
\newblock {\em J. Stat. Phys.}, 154(1-2):334--355, 2014.

\bibitem{Tran08}
V.~C. Tran.
\newblock Large population limit and time behaviour of a stochastic particle
  model describing an age-structured population.
\newblock {\em ESAIM Probab. Stat.}, 12:345--386, 2008.

\end{thebibliography}

\end{document}